\documentclass[12pt,reqno]{amsart}

\usepackage{epsfig}
\usepackage{amsmath, amsthm, amsfonts}
\usepackage{graphicx}

\numberwithin{equation}{section}

\newcommand{\R}{{\mathbb R}}
\newcommand{\C}{{\mathbb C}}
\newcommand{\N}{{\mathbb N}}

\newcommand{\Z}{{\mathbb Z}}

\renewcommand{\Re}{{\operatorname{Re\,}}}
\renewcommand{\Im}{{\operatorname{Im\,}}}

\newcommand{\Ai}{{\operatorname{Ai}}}
\newcommand{\Bi}{{\operatorname{Bi}}}

\newcommand{\I}{{\bold I}}

\newcommand{\al}{\alpha}
\newcommand{\be}{\beta}
\newcommand{\ga}{\gamma}

\newcommand{\supp}{\textrm{supp} \ }
\newcommand{\ep}{\varepsilon}
\newcommand{\de}{\delta}
\newcommand{\De}{\Delta}

\newcommand{\sg}{\sigma}
\newcommand{\Sg}{\Sigma}
\newcommand{\om}{\omega}
\newcommand{\Om}{\Omega}
\renewcommand{\th}{\theta}

\newcommand{\acal}{{\mathcal A}}
\newcommand{\mcal}{{\mathcal M}}

\newtheorem{theo}{{\sc \bf Theorem}}[section]

\newtheorem{lem}[theo]{{\sc \bf Lemma}}
\newtheorem{prop}[theo]{{\sc \bf Proposition}}

\newenvironment{rem}{\medskip\noindent{\it Remark:\/} }{\medskip}
\newenvironment{defin}{\medskip\noindent{\it Definition:\/} }{\medskip}

\begin{document}

\title[Uniform Asymptotics for Discrete Orthogonal Polynomials]
{Uniform Asymptotics for Discrete Orthogonal Polynomials with respect to Varying Exponential Weights on a Regular Infinite Lattice}

\thanks{The first author is supported in part
by the National Science Foundation (NSF) Grant DMS-0652005.}

\thanks{Both authors would like to thank the referee for his or her very important comment.}

\author{Pavel Bleher}
\address{Department of Mathematical Sciences,
Indiana University-Purdue University Indianapolis,
402 N. Blackford St., Indianapolis, IN 46202, U.S.A.}
\email{bleher@math.iupui.edu}

\author{Karl Liechty}
\address{Department of Mathematical Sciences,
Indiana University-Purdue University Indianapolis,
402 N. Blackford St., Indianapolis, IN 46202, U.S.A.}
\email{kliechty@math.iupui.edu}

\date{\today}

\begin{abstract} 
We consider the large $N$ asymptotics of a system of discrete orthogonal polynomials on an infinite regular lattice of mesh $\frac{1}{N}$, with weight $e^{-NV(x)}$, where $V(x)$ is a real analytic function with sufficient growth at infinity.  The proof is based on formulation of an interpolation problem for discrete orthogonal polynomials, which can be converted to a Riemann-Hilbert problem, and steepest descent analysis of this Riemann-Hilbert problem.
\end{abstract}

\maketitle

\section{Introduction}
For a given $N\in\N$, introduce the regular infinite lattice,
\begin{equation}\label{in1}
L_N=\left\{ x_{k,N}=\frac{k}{N}\,,\; k\in\Z\right\}.
\end{equation}
We consider polynomials orthogonal on $L_N$ with respect to the varying exponential weight
\begin{equation}\label{in2}
w_N(x)=e^{-NV(x)},
\end{equation}
where $V(x)$ is a real analytic function such that, for some $\ep >0$, $V$ has analytic extension into the strip
\begin{equation}
|\,\Im z| < \ep,
\end{equation}
and satisfies the growth condition
\begin{equation}\label{in3}
\frac{\Re V(z)}{\log(|z|^2+1)} \to + \infty \quad \textrm{as} \ |z| \to \infty, \ |\Im z| < \ep.
\end{equation}
More specifically, we introduce the system of monic orthogonal polynomials,
\[ 
P_n(x)=x^n+p_{n,n-1}x^{n-1}+\ldots+p_{n0},\quad n=0,1,\ldots,
\]
such that
\begin{equation}\label{in4}
\sum_{x\in L_N} P_m(x)P_n(x)w_N(x)=h_n\delta_{mn}\,,
\end{equation}
for some normalizing coefficients $h_n$.
Existence and uniqueness of this system of orthogonal polynomials is guaranteed by condition (\ref{in3}).
These orthogonal polynomials satisfy the three-term recurrence relation
\begin{equation}\label{in5}
xP_n(x)=P_{n+1}(x)+\be_n P_n(x)+\ga_{n}^2 P_{n-1}(x).
\end{equation}
We will explore the asymptotics of the quantities $\ga_n$, $\be_n,$ and $h_n$ for $n=N,N-1$, 
as well as pointwise asymptotics of the polynomials $P_N(x)$ as $N \to \infty$. 

The present work has the three predecessors:
\begin{enumerate}
  \item the work \cite{DKMVZ} of Deift, Kriecherbauer, McLaughlin, Venakides, and Zhou,  
in which the large $N$ asymptotics has been obtained for orthogonal polynomials with respect 
to varying exponential weights on the real line,
  \item the work \cite{BKMM} of Baik, Kriecherbauer, McLaughlin, and Miller, 
in which the large $N$ asymptotics has been obtained for orthogonal polynomials with respect 
to varying exponential weights on a lattice in a finite interval, and
  \item the work \cite{BL2} of Bleher and Liechty, 
in which the large $N$ asymptotics has been obtained for orthogonal polynomials with respect 
to the varying exponential weight $w_N(x)=e^{-N(|x|-\zeta x)}$ on the infinite lattice $L_N$.
\end{enumerate}
Also, a very important ingredient comes from the work \cite{Kui} of 
Kuijlaars, in which analytic properties of equilibrium measures with constraints are established. 
  
The asymptotic analysis of the polynomials $P_N(x)$ in this work 
will be based on the Interpolation Problem 
for discrete orthogonal polynomials, which is introduced in the work \cite{BoB} of 
Borodin and Boyarchenko (see also \cite{BKMM}, \cite{BL1}, \cite{BL2}). The asymptotic analysis
of $P_N(x)$ will consist of three steps. The first step will be a reduction of the Interpolation
Problem to a Riemann-Hilbert problem on a contour on the complex plane, which we accomplish following the general approach introduced in the paper \cite{M} of Miller and in the monograph \cite{KMM} of Kamvissis, McLaughlin, and Miller.  The second step will
be an application of the nonlinear steepest descent method of Deift and Zhou \cite{DZ} to
the Riemann-Hilbert problem under consideration, and the third and final step will be
a derivation of the asymptotic formulae both for the orthogonal polynomials $P_N(x)$
and for the recurrence coefficients. To apply the nonlinear steepest descent method 
to the orthogonal polynomials $P_N(x)$ we need to study the corresponding equilibrium measure.

\section{Equilibrium Measure}\label{equilibrium}
The significance of the equilibrium measure is that, as we will see, it gives the 
limiting distribution of zeros of the polynomial $P_N(x)$.
By definition, the equilibrium measure is a solution to a variational problem.  
Namely, let us consider the following set of probability measures on $\R^1$:
\begin{equation}\label{eq1}
\mcal=\{0\leq \nu\leq \sg,\;\;\nu(\R^1)=1\},
\end{equation}
where $\sg$ is the Lebesgue measure, and let us introduce the functional
\begin{equation}\label{eq2}
H(\nu)=\iint \log \frac{1}{|x-y|}d\nu (x)d\nu (y)+\int V(x)d\nu(x),\qquad \nu\in\mcal.
\end{equation}
The equilibrium measure minimizes this functional over some set of measures. 
In the case of continuous orthogonal polynomials, we minimize over the set 
of probability measures on the real line.  However, in the case of discrete 
orthogonal polynomials, we must introduce the upper constraint, $\nu\le\sg$, 
in order to account for an interlacing property of the zeroes of orthogonal polynomials.

It is a general fact, (see, e.g. \cite{Sze}) that for any system of polynomials orthogonal on the real line 
with respect to a real weight, the $n$th polynomial has $n$ real distinct distinct zeroes.  
Furthermore, the zeroes of a system of discrete orthogonal polynomials satisfy 
an interlacing property with regard to the location of the nodes of the lattice $L_N$, 
so that no more than one zero may lie between any pair of adjacent nodes.  
It therefore follows that, if we denote by $\mu_N$ the normalized counting measure 
on the zeroes of the $N$th orthogonal polynomial in our system,
\begin{equation}\label{eq3}
\mu_N(a,b)\leq b-a + \frac{1}{N} \quad \textrm{for any} \quad -\infty<a<b<\infty,
\end{equation}
so that $\mu\le\sg$, where $\mu = \lim_{N\to \infty} \mu_N$. 
With this constraint in mind, we define
\begin{equation}\label{eq4}
E_0=\inf_{\nu\in\mcal} H(\nu).
\end{equation}
It is possible to prove that there exists a unique minimizer $\nu_0$, so that  
\begin{equation}\label{eq5}
E_0=H(\nu_0),
\end{equation}
see, e.g., the works of Saff and Totik \cite{ST}, Dragnev and Saff \cite{DS} and Kuijlaars \cite{Kui}.
The minimizer is called the {\it equilibrium measure}. 

The equilibrium measure $\nu_0$ is uniquely determined by the {\it Euler-Lagrange variational conditions}:
there exists a {\it Lagrange multiplier} $l$ such that
\begin{equation}\label{eq6}
2\int \log|x-y| d\nu_0 (y) - V(x) \left\{
\begin{aligned}
&\geq l \quad \textrm{for}\quad x \in \supp \nu_0\\
&\leq l \quad \textrm{for}\quad x \in \supp (\sg-\nu_0),
\end{aligned}\right.
\end{equation}
see the works \cite{DM} of Deift and McLaughlin and \cite{DS}.
In particular,
\begin{equation}\label{eq7}
2\int \log|x-y| d\nu_0 (y) - V(x)=l
  \quad \textrm{for}\quad x \in \supp \nu_0\cap \supp (\sg-\nu_0).
\end{equation}
The equilibrium measure $\nu_0$ possesses a number of nice analytical properties, as shown by Kuijlaars
in \cite {Kui}. We will use these analytic properties, so let us discuss the results of \cite{Kui}.

First, observe that the constraint $\nu_0\le\sg$ implies the existence of the density,
\begin{equation}\label{eq8}
\rho(x)=\frac{d\nu_0}{dx}\,.
\end{equation}
We can partition $\R$ into the three sets
\begin{equation}\label{eq9}
\begin{aligned}
I^0&=\{x \in \R : 2\int \log|x-y| d\nu_0 (y) - V(x) = l\} \\
I^+&=\{x \in \R : 2\int \log|x-y| d\nu_0 (y) - V(x) > l\} \\
I^-&=\{x \in \R : 2\int \log|x-y| d\nu_0 (y) - V(x) < l\}. \\
\end{aligned}
\end{equation}
The structure of the equilibrium measure is well described in the following theorem of Kuijlaars, obtained in \cite{Kui}.
\begin{theo} (Kuijlaars)  For any real analytic potential $V(x)$ satisfying (\ref{in3}), the following hold:
\begin{enumerate}
  \item The density $\rho(x)$ of the constrained equilibrium measure $\nu_0$ (defined in (\ref{eq5})) is continuous.
  \item The sets $I^+$ and $I^-$ are both finite unions of open intervals.
  \item The density $\rho$ is real analytic on the open set $\{x : 0 < \rho(x) < 1 \}$.
  \item The density $\rho$ has the representation
  \begin{equation}\label{eq10}
\rho(x)=\frac{1}{\pi}\sqrt{q_1^+(x)} \quad \textrm{for} \quad x \in I^0 \cup I^-\,,
\end{equation}
where $q_1^+$ is the positive part of a function $q_1$ defined on $I^0 \cup I^-$ 
which is real analytic on the interior of $I^0 \cup I^-$.  The function $q_1$ is negative on $I^-$, so that 
\begin{equation}\label{eq11}
\rho(x)=0 \quad \textrm{for} \quad x \in  I^-\,,
\end{equation}
and it is nonnegative on $I^0$, so that
\begin{equation}\label{eq12}
\rho(x)=\frac{1}{\pi}\sqrt{q_1(x)} \quad \textrm{for} \quad x \in I^0 .
\end{equation}
  \item The density $\rho$ has the representation
  \begin{equation}\label{eq13}
\rho(x)=1 - \frac{1}{\pi}\sqrt{q_2^+(x)} \quad \textrm{for} \quad x \in I^0 \cup I^+\,,
\end{equation}
where $q_2^+$ is the positive part of a function $q_2$ defined on $I^0 \cup I^+$ 
which is real analytic on the interior of $I^0 \cup I^+$.  The function $q_2$ is negative 
on $I^+$, so that 
\begin{equation}\label{eq14}
\rho(x)=1 \quad \textrm{for} \quad x \in  I^+\,,
\end{equation}
and it is nonnegative on $I_0$, so that
\begin{equation}\label{eq15}
\rho(x)=1-\frac{1}{\pi}\sqrt{q_2(x)} \quad \textrm{for} \quad x \in I^0.
\end{equation}
\end{enumerate}
\end{theo}

\begin{rem} It follows from equations (\ref{eq12}) and (\ref{eq15}) that
\begin{equation}\label{eq16}
\frac{1}{\pi}\sqrt{q_1(x)}=1-\frac{1}{\pi}\sqrt{q_2(x)} \quad \textrm{for} \quad x \in I^0,
\end{equation}
hence $q_1$ and $q_2$ uniquely determine each other.
\end{rem}

Notice that, according to point (2) of this theorem, 
the connected components of $I^0$ are either closed intervals or isolated points.  
Since $\nu_0$ has compact support, we can write
\begin{equation}\label{eq17}
I^0=\bigsqcup_{j=1}^q [\al_j,\be_j]
\end{equation}
where
\begin{equation}\label{eq18}
\begin{aligned}
\al_j &\leq \be_j \quad &\textrm{for} \quad j&=1,\dots,q \\
\be_j&<\al_{j+1} \quad &\textrm{for} \quad j&=1,\dots,q-1.
\end{aligned}
\end{equation}
Notice that the intervals $(-\infty,\al_1)$ and $(\be_q,\infty)$ are components of $I^-$.  
The interval $(\be_j,\al_{j+1})$ for $1\leq j<q$ is a component of either $I^+$ or $I^-$.  
We therefore adopt the notation
\begin{equation}\label{eq19}
\begin{aligned}
\acal_v&=\bigg\{j\in \{1,\dots,q-1\} : (\be_j,\al_{j+1}) \subset I^-\bigg\} \\
\acal_s&=\bigg\{j\in \{1,\dots,q-1\} : (\be_j,\al_{j+1}) \subset I^+\bigg\}.
\end{aligned}
\end{equation}
We will call an equilibrium measure $\nu_0$ {\it regular} if the following hold:
\begin{enumerate}
  \item $q_1$ and $q_2$ are non-vanishing on the interior of $I^0$.
  \item $I^0$ contains no isolated points, so that $\al_j<\be_j$ for all $j=1,\ldots,q$.
  \item If $j \in \acal_v$, then $q_1'(\be_j) \neq 0 $ and $q_1'(\al_{j+1}) \neq 0$. Also, $q_1'(\al_1) \neq 0 $ and $q_1'(\be_q) \neq 0$.
  \item If $j \in \acal_s$, then $q_2'(\be_j) \neq 0 $ and $q_2'(\al_{j+1}) \neq 0$.
\end{enumerate}
For the remainder of this paper, we will assume that our equilibrium measure is regular.  In this case, the sets $I^0$, $I^+$ and $I^-$ are each finite unions of intervals, so that 
\begin{equation}\label{eq20}
-\infty<\al_1<\be_1<\al_2<\be_2<\cdots<\al_q<\be_q<\infty,
\end{equation}
and we classify these intervals as follows:

\begin{defin}A {\it void} is an open subinterval $(\be_j,\al_{j+1})$, $j\in\acal_v$, or one of the intervals $(-\infty, \al_1)$, $(\be_q,\infty)$.  The union of all voids is $I^-$.
\end{defin}

\begin{defin}A {\it saturated region} is an open subinterval $(\be_j,\al_{j+1})$, $j\in\acal_s$.  The union of all saturated regions is $I^+$.
\end{defin}

\begin{defin}A {\it band} is an open subinterval $(\al_j,\be_j)$, $j=1,\ldots,q$.  
The union of all bands is the interior of $I^0$.
\end{defin}

\begin{figure}
\vskip 1cm
\scalebox{0.7}{\includegraphics{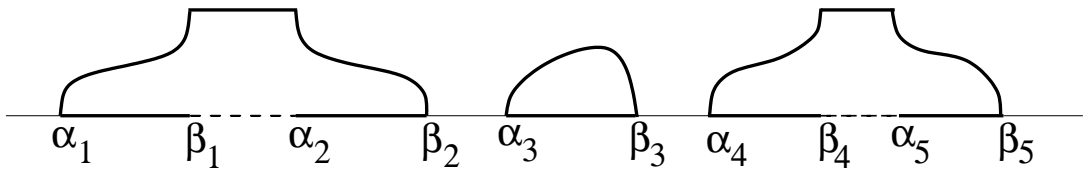}}
\caption{The graph of the density function for a hypothetical equilibrium measure with $q=5$.  Bands are denoted by bold segments, saturated regions by dashed segments, and voids by thin segments.}
\label{fig1}
\end{figure}

Observe that $\rho(x)=0$ on any void $(\be_j,\al_{j+1})$, $\rho(x)=1$ on any saturated
interval $(\be_j,\al_{j+1})$, and $0<\rho(x)<1$ on any band $(\al_j,\be_j)$, see Figure \ref{fig1}.
In addition, at the end-points of any band, $\rho(x)$ has a square-root singularity. Namely,
if $\al_j$ is a common end-point of a band and a void then as $x\to+0$,
\begin{equation}\label{eq21}
\begin{aligned}
\rho(\al_j+x)=C\sqrt x\,(1+O(x)),\qquad C=|q_1'(\al_j)|^{1/2}>0.
\end{aligned}
\end{equation}
and 
if $\al_j$ is a common end-point of a band and a saturated region then as $x\to+0$,
\begin{equation}\label{eq22}
\begin{aligned}
\rho(\al_j+x)=1-C\sqrt x\,(1+O(x)),\qquad C=|q_2'(\al_j)|^{1/2}>0.
\end{aligned}
\end{equation}
Similarly, if $\be_j$ is a common end-point of a band and a void then as $x\to+0$,
\begin{equation}\label{eq23}
\begin{aligned}
\rho(\be_j-x)=C\sqrt x\,(1+O(x)),\qquad C=|q_1'(\be_j)|^{1/2}>0.
\end{aligned}
\end{equation}
and 
if $\be_j$ is a common end-point of a band and a saturated region then as $x\to+0$,
\begin{equation}\label{eq24}
\begin{aligned}
\rho(\be_j-x)=1-C\sqrt x\,(1+O(x)),\qquad C=|q_2'(\be_j)|^{1/2}>0.
\end{aligned}
\end{equation}

In the next section we introduce the $g$-function, which will be our means of exploiting the equilibrium measure.

\section{The $g$-function}
Define the $g$-function on $\C \setminus (-\infty, \be_q]$ as 
\begin{equation}\label{g1}
g(z)=\int_{\al_1}^{\be_q} \log(z-x)d\nu_0(x)
\end{equation}
where we take the principal branch for the logarithm.  Also, introduce the numbers $\Om_j$ for $j=1,\dots,q-1$ as
\begin{equation}\label{eq25a}
\begin{aligned}
\Omega_j&=\left\{
\begin{aligned}
&2\pi \int_{\al_{j+1}}^{\be_q} \rho(x)dx \quad \textrm{for} \quad j\in \acal_v \\
&2\pi \int_{\al_{j+1}}^{\be_q} \rho(x)dx+2\pi\al_{j+1} \quad \textrm{for} \quad j\in \acal_s.
\end{aligned}\right. \\
\end{aligned}
\end{equation}

{\bfseries Properties of \( g(z)\):}
\begin{enumerate}
\item \( g(z)\) is analytic in \(\C\setminus
    (-\infty,\be_q]\).
\item For large $z$,
\begin{equation}\label{g2}
 g(z)=\log z-\sum_{j=1}^\infty \frac{g_j}{z^j}\,,\qquad 
g_j=\int_{\al_1}^{\be_q} \frac{x^j}{j}\,d\nu_0(x).
\end{equation}
\item 
\begin{equation}\label{g3}
 g'(z)=\int_\R \frac{\rho(x)dx}{z-x}
\end{equation}
is the resolvent of the equilibrium measure.
\item From (\ref{eq9}), we have that 
\begin{equation}\label{g5}
g_+(x)+g_-(x)\,
\left\{
\begin{aligned}
&=V(x)+l \quad \textrm{for}\quad x \in I^0 \\
&> V(x)+l \quad \textrm{for}\quad x \in I^+ \\
&< V(x)+l \quad \textrm{for}\quad x \in I^-,
\end{aligned}
\right.
\end{equation}
where $g_+$ and $g_-$ refer to the limiting values from the upper and lower half-planes, respectively.

\item Equation (\ref{g1}) implies that the function
\begin{equation}\label{g6}
 G(x)\equiv g_+(x)-g_-(x)
\end{equation}
is pure imaginary for all real 
\( x\), and 
\begin{equation}\label{g7}
G(x)=2\pi i\int_x^{\be_q}\rho(s)\,ds.
\end{equation} 
Thus
\begin{equation}\label{g8}
G(x)=\left\{
\begin{aligned}
&i\Omega_j \quad \textrm{for} \quad \be_j<x<\al_{j+1}, \quad \textrm{and} \quad j\in \acal_v \\
&i\Omega_j-2\pi ix \quad \textrm{for} \quad \be_j<x<\al_{j+1}, \quad \textrm{and} \quad j\in \acal_s.
\end{aligned}
\right.
\end{equation}
From (\ref{g5}) and (\ref{g7}) we obtain that
\begin{equation}\label{g9}
 2g_{\pm}(x)=
V(x)+l\pm 2\pi i \int_x^{\be_q} \rho(s)ds\quad\textrm{for} \quad x \in I^0.
\end{equation}
\item 
Also, from (\ref{g7}), we get that $G(x)$ is real analytic on the sets $I^+$, $I^-$, and on the interior of $I^0$.  We can therefore extend $G$ into a complex neighborhood of any interval of analyticity for $\rho$, and the Cauchy-Riemann equations imply that 
\begin{equation}\label{g10}
 \left. \frac{dG(x+iy)}{dy}\right|_{y=0}=2\pi \rho(x)\ge 0.
\end{equation}
\end{enumerate}
Observe that from (\ref{g5}) we have that
\begin{equation}\label{g11}
 G(x)=2g_+(x)-V(x)-l=-[2g_-(x)-V(x)-l],\quad x \in I^0.
\end{equation}

\medskip

\section{Main Results}\label{results}
In this section, we summarize the main results of the paper.  In order to do so, we must first introduce some notations.  
Introduce the numbers $\Om_{j,N}$ for $j=0,\dots,q$ as
\begin{equation}\label{eq25}
\begin{aligned}
\Om_{j,N}&=\left\{
  \begin{aligned}
&N\Om_j \quad \textrm{for} \quad j\in \acal_v  \\
&\pi+N\Omega_j \quad \textrm{for} \quad j\in \acal_s \\
&2\pi N \quad \textrm{for} \quad j=0 \\
&0 \quad \textrm{for} \quad j=q.
\end{aligned}\right.
\end{aligned}
\end{equation}
and the vector
\begin{equation}\label{eq26}
\Om_N=(\Om_{1,N},\dots,\Om_{q-1,N}).
\end{equation}
Let 
\begin{equation}\label{m3}
R(z)\equiv \prod_{j=1}^q (z-\al_j)(z-\be_j)
\end{equation}
and let $X$ be the two-sheeted Riemann surface of genus $g\equiv q-1$ associated with $\sqrt {R(z)}$ with cuts on the intervals $(\al_j,\be_j)$.  We fix the first sheet of $X$ by the condition
\begin{equation}\label{m4}
\sqrt {R(z)} > 0 \quad \textrm{for} \quad z > \be_q
\end{equation}
on the first sheet.

Introduce the following homology basis on $X$.  For any $j \in \{1,\cdots,q-1\}$, let $A_j$ be a cycle enclosing the interval $(\be_j,\al_{j+1})$ (passing through the intervals $(\al_j,\be_j)$ and $(\al_{j+1},\be_{j+1})$), oriented clockwise, such that the piece of $A_j$ which lies in the upper half-plane also lies on the first sheet of $X$, while the piece of $A_j$ which lies in the lower half plane also lies on the second sheet of $X$.  Also for any $j \in \{1,\cdots,q-1\}$, let $B_j$ be a cycle enclosing the interval $(\al_1,\be_j)$ (passing through the intervals $(-\infty,\al_1)$ and $(\be_j,\al_{j+1})$), oriented clockwise, and lying entirely on the first sheet of $X$.  Then the cycles $(A_1,\dots,A_{q-1},B_1,\dots.B_{q-1})$ form a canonical homology basis for $X$.

Now consider the the $g$-dimensional complex linear space $\Om$ of holomorphic one-forms on $X$,
\begin{equation}\label{m5}
\Om=\left\{\om=\sum_{j=0}^{q-2}\frac{c_j z^j dz}{\sqrt{R(z)}}\right\},
\end{equation}
and the basis 
\begin{equation}\label{m6}
\om=(\om_1,\dots,\om_{q-1})
\end{equation}
normalized such that 
\begin{equation}\label{m7}
\int_{A_j} \om_k = \delta_{jk}.
\end{equation}
Notice that the basis $\om$ is real.  That is, for the basis elements
\begin{equation}\label{m8}
\om_j=\sum_{k=1}^{q-1}\frac{c_{jk} z^{k-1} dz}{\sqrt{R(z)}},
\end{equation}
the coefficients $c_{jk}$ are real.

Now define the associated matrix of $B$-periods as
\begin{equation}\label{m9}
\tau=(\tau_{jk}), \qquad \tau_{jk}=\int_{B_j} \om_k, \qquad j,k=1,\dots,q-1.
\end{equation}
Since $\sqrt{R(z)}$ is pure imaginary on the intervals $(\al_j,\be_j)$, the numbers $\tau_{jk}$ are pure imaginary.  Furthermore, the matrix $\tau$ is symmetric and the matrix $-i\tau$ is positive definite (see \cite{FK}).

We now define the Riemann theta function associated with $\tau$ as
\begin{equation}\label{m10}
\th(s)=\sum_{m\in \Z^g} e^{2\pi i (m,s)+\pi i (m,\tau m)}, \qquad s\in \C^g,
\end{equation}
where $(m,s)=\sum_{j=1}^{q-1} m_j s_j$.  Because the quadratic form $i(m,\tau m)$ is negative definite, the sum in (\ref{m10}) is absolutely convergent for all $s \in \C^g$, and thus $\th(s)$ is an entire function in $\C^g$.  Notice that the theta function is an even function and satisfies the periodicity properties
\begin{equation}\label{m11}
\th(s+e_j)=\th(s), \qquad \th(s+\tau_j)=e^{-2\pi i s_j -\pi i \tau_{jj}}\th(s)
\end{equation}
where $e_j=(0,\dots,1,\dots,0)$ is the $j^{th}$ canonical basis vector in $\C^g$, and $\tau_j=\tau e_j$.

Introduce now the vector valued function
\begin{equation}\label{m12}
u(z)=\int_{\be_q}^z \om, \quad \textrm{for} \quad z\in \C\setminus (\al_1,\be_q),
\end{equation}
where $\om=(\om_1,\dots,\om_g)$ is defined in (\ref{m6}), and the contour of integration lies in $\C\setminus (\al_1,\be_q)$ on the first sheet of $X$.  Notice that $u(z)$ is well defined as a function with values in $\C^g / \Z^g$ except on the interval $(\al_1,\be_q)$, where it takes limiting values from the upper and lower half-planes.

Introduce also the function
\begin{equation}\label{m2b}
\ga(z)=\prod_{j=1}^q\left(\frac{z-\al_j}{z-\be_j}\right)^{1/4}
\end{equation}
with cuts on $I^0$, taking the branch such that $\ga(z) \sim 1$ as $z\to \infty$.  It can be seen that, on the first sheet of $X$, the function $\ga-\frac{1}{\ga}$ has exactly one zero in each of the intervals $(\be_j,\al_{j+1})$, and is non-zero elsewhere, and that the function $\ga+\frac{1}{\ga}$ has no zeroes on the first sheet of $X$.  Define the numbers $x_j$ as
\begin{equation}\label{m2c}
x_j \in (\be_j,\al_{j+1}), \quad \ga(x_j)-\frac{1}{\ga(x_j)}=0.
\end{equation}
Define the vector of Riemann constants 
\begin{equation}\label{m24}
K \equiv -\sum_{j=1}^{q-1} u(\be_j)
\end{equation}
and the vector 
\begin{equation}\label{m25}
d \equiv -K+\sum_{j=1}^{q-1} u(x_j).
\end{equation}
Then
\begin{equation}\label{m26}
\th(u(x_j)-d)=0 \quad \textrm{for} \quad  j\in\{1,\dots,q-1\},
\end{equation}
and $\{x_j\}_{j=1}^g$ give all the zeroes of the function $\th(u(z)-d)$.  In addition, the function $\th(u(z)+d)$ has no zeroes on the first sheet of $X$. 

Finally, for $j=1,\dots q$, introduce the functions
\begin{equation}\label{m27}
\psi_{\al_j}(z)=-\left\{\frac{3\pi}{2}\int_{\al_j}^z \rho(t)dt\right\}^{2/3}, \quad \psi_{\be_j}(z)=-\left\{\frac{3\pi}{2}\int_z^{\be_j} \rho(t)dt\right\}^{2/3},
\end{equation}
and the functions
\begin{equation}\label{m28}
\begin{aligned}
\mcal_1(z)&=\frac{\th(u(\infty)+d)}{\th(u(\infty)+\frac{\Om_N}{2\pi}+d)}\frac{\ga(z)+\ga(z)^{-1}}{2}\frac{\th(u(z)+\frac{\Om_N}{2\pi}+d)}{\th(u(z)+d)} \\
\mcal_2(z)&=\frac{\th(u(\infty)+d)}{\th(u(\infty)+\frac{\Om_N}{2\pi}+d)}\frac{\ga(z)-\ga(z)^{-1}}{2}\frac{\th(u(z)-\frac{\Om_N}{2\pi}-d)}{\th(u(z)-d)}.
\end{aligned}
\end{equation}
Notice that $\mcal_1$ and $\mcal_2$ depend quasiperiodically on $N$, thus are $O(1)$ as $N \to \infty$.  

The asymptotics of the normalizing constants in equation (\ref{in4}) and of the recurrence coefficients in equation (\ref{in5}) are presented in the following theorem.

\begin{theo}\label{coeff}(Asymptotics of recurrence coefficients) Let $V(x)$ be a real analytic function satisfying (\ref{in3}) which yields a regular equilibrium measure (\ref{eq5}), and let $\{P_n\}_{n=0}^\infty$ be the system of orthogonal polynomials defined according to (\ref{in4}).  Then as $N \to \infty$, the normalizing constants in (\ref{in4}) and recurrence coefficients in (\ref{in5}) admit the following asymptotic expansions.  

\begin{equation}\label{mr1} 
h_N = \frac{N\pi}{2}e^{Nl}\left(\sum_{j=1}^q(\be_j-\al_j)\right)\frac{\th(u(\infty)+d)\th(u(\infty)-\frac{\Om_N}{2\pi}-d)}{\th(u(\infty)-d)\th(u(\infty)+\frac{\Om_N}{2\pi}+d)}\left[1+O\left(\frac{1}{N}\right)\right],
\end{equation}
\begin{equation}\label{mr2} 
h_{N-1} = 8N\pi e^{Nl}\left(\sum_{j=1}^q(\be_j-\al_j)\right)^{-1}\frac{\th(u(\infty)-d)\th(u(\infty)-\frac{\Om_N}{2\pi}+d)}{\th(u(\infty)+d)\th(u(\infty)+\frac{\Om_N}{2\pi}-d)}\left[1+O\left(\frac{1}{N}\right)\right],
\end{equation}
\begin{equation}\label{mr3} 
\ga_N^2 = \left(\frac{\sum_{j=1}^q(\be_j-\al_j)}{4}\right)^2\frac{\th(u(\infty)+d)^2\th(u(\infty)-\frac{\Om_N}{2\pi}-d)\th(u(\infty)+\frac{\Om_N}{2\pi}-d)}{\th(u(\infty)-d)^2\th(u(\infty)+\frac{\Om_N}{2\pi}+d)\th(u(\infty)-\frac{\Om_N}{2\pi}+d)}+O\left(\frac{1}{N}\right),
\end{equation}
\begin{equation}\label{mr4}
\begin{aligned}
\be_{N-1}&=\frac{\sum_{j=1}^q(\be_j^2-\al_j^2)}{2\sum_{j=1}^q(\be_j-\al_j)}+\left(\frac{\nabla \th(u(\infty)+\frac{\Om_N}{2\pi}-d)}{\th(u(\infty)+\frac{\Om_N}{2\pi}-d)}-\frac{\nabla \th(u(\infty)+\frac{\Om_N}{2\pi}+d)}{\th(u(\infty)+\frac{\Om_N}{2\pi}+d)} \right.\\
&\qquad \left.+\frac{\nabla \th(u(\infty)+b)}{\th(u(\infty)+b)}-\frac{\nabla \th(u(\infty)-d)}{\th(u(\infty)-d)},u'(\infty)\right)+O\left(\frac{1}{N}\right).
\end{aligned}
\end{equation}
where $\nabla \th$ is the gradient of $\th$,
\begin{equation}
u'(\infty)=(c_{1,q-1},c_{2,q-1},\dots,c_{q-1,q-1}),
\end{equation}
and the numbers $c_{jk}$ are defined in (\ref{m8}).
\end{theo}
Notice that, up to the lattice scaling factor $N$ in the normalizing coefficients, these asymptotics are similar to the results obtained in \cite{DKMVZ} for continuous orthogonal polynomials.

The remaining theorems in this section present pointwise asymptotics of the polynomials $P_N(z)$ in various regions of the real line and complex plane.
\begin{theo}\label{voids}(Asymptotics of $P_N(z)$ in voids)  Let $K\subset\C$ be a
compact set on the complex plane such that $K$ does not intersect with the support of the equilibrium measure $\nu_0$.
Then for any $z\in K$, we have that
\begin{equation}\label{mr5}
P_N(z)=e^{Ng(z)}\left[\mcal_1(z)+O(N^{-1})\right].   
\end{equation}
The error term $O(N^{-1})$ is uniform in $K$.
\end{theo}

The function $e^{Ng(z)} \mcal_1(z)$ is analytic in a neighborhood of any compact subset of any void, thus this formula gives asymptotics of $P_N(x)$ for $x$ in a void.  In particular, notice that this function has no zeroes in the exterior intervals $(-\infty,\al_1)$ and $(\be_q,\infty)$, and at most one zero in any other void.

\begin{theo}\label{bands}(Asymptotics of $P_N(z)$ in bands)  Let $K$ be a compact
subset of the interior of $I^0$. Then for any point $x\in K$, we have that
\begin{equation}\label{mr6a}
P_N(x)=2e^{\frac{N}{2}(V(x)+l)}\left[\Re \left(e^{iN\pi \phi(x)} \mcal_{1+}(x)\right)+O(N^{-1})\right],
\end{equation}
where $\mcal_{1+}(x)$ refers to the limiting value of the function $\mcal_{1}(z)$ from the upper half-plane, and
\begin{equation}\label{mr7}
\phi(x):=\int_x^{\be_q} \rho(t)dt.
\end{equation}
The error term $O(N^{-1})$ is uniform in $K$.
\end{theo}

\begin{theo}\label{satreg}(Asymptotics of $P_N(z)$ in saturated regions)
 Let $K$ be a compact
subset of $I^+$.
Then the exists $\ep >0$ such that, for any point $x\in K$, we have that
\begin{equation}\label{mr8}
P_N(x)=e^{NL(x)}\left[2\sin(N\pi x)\left(\Im \left(e^{\frac{iN\Om_{j}}{2}} \mcal_{1+}(x)\right)+O(N^{-1})\right)+O(e^{-N\ep})\right],
\end{equation}
where $\mcal_{1+}(x)$ refers to the limiting value of the function $\mcal_{1}(z)$ from the upper half-plane, and
\begin{equation}\label{mr9}
L(x):=\int_{\al_1}^{\be_q} \log|x-t|\rho(t)dt.
\end{equation}
Both of the error terms, $O(N^{-1})$ and $O(e^{-N\ep})$, are uniform in $K$.
\end{theo}

The remaining theorems in this section use the Airy functions $\Ai$ and $\Bi$ (see, e.g. \cite{Ol}).

\begin{theo}\label{band_void}(Asymptotics of $P_N(z)$ at band-void edge points)
Let $j \in \mathcal{A}_v\cup\{ q\}$, so that the point $\be_j$ is the right endpoint of a band 
and the left endpoint of a void. Then there exists $\ep>0$ such that, for $|z-\be_j| < \ep$,
\begin{equation}\label{mr10}
\begin{aligned}
P_N(z)&=e^{\frac{N}{2}(V(z)+l)}\left\{N^{1/6}\psi_{\be_j}(z)^{1/4}\Ai(N^{2/3}\psi_{\be_j}(z))\left[e^{\pm  \frac{i\Om_{j,N}}{2}}\mcal_1(z)+e^{\mp  \frac{i\Om_{j,N}}{2}}\mcal_2(z)+O(N^{-1})\right]\right. \\
&\quad \left.- N^{-1/6}\psi_{\be_j}(z)^{-1/4}\Ai ' (N^{2/3}\psi_{\be_j}(z))\left[e^{\pm  \frac{i\Om_{j,N}}{2}}\mcal_1(z)-e^{\mp  \frac{i\Om_{j,N}}{2}}\mcal_2(z)+O(N^{-1})\right]\right\}
\end{aligned}
\end{equation}
for $\pm \Im z >0$.

Let $j \in \mathcal{A}_v\cup\{ 0\}$, so that the point $\al_{j+1}$ is the left endpoint of a band and the right endpoint of a void.  There exists $\ep>0$ such that, for $|z -\al_{j+1}|<\ep$,
\begin{equation}\label{mr12}
\begin{aligned}
P_N(z)&=e^{\frac{N}{2}(V(z)+l)}\left\{N^{1/6}\psi_{\al_{j+1}j}(z)^{1/4}\Ai(N^{2/3}\psi_{\al_{j+1}}(z))\left[e^{\pm  \frac{i\Om_{j,N}}{2}}\mcal_1(z)-e^{\mp  \frac{i\Om_{j,N}}{2}}\mcal_2(z)+O(N^{-1})\right]\right. \\
&\quad \left. - N^{-1/6}\psi_{\al_{j+1}}(z)^{-1/4}\Ai ' (N^{2/3}\psi_{\al_{j+1}}(z))\left[e^{\pm  \frac{i\Om_{j,N}}{2}}\mcal_1(z)+e^{\mp  \frac{i\Om_{j,N}}{2}}\mcal_2(z)+O(N^{-1})\right]\right\}
\end{aligned}
\end{equation}
for $\pm \Im z >0$.
\end{theo}

\begin{theo}\label{band_sat-reg}(Asymptotics of $P_N(z)$ at band-saturated region edge points)
Let $j \in \mathcal{A}_s$.  Then the point $\be_j$ is the right endpoint of a band and the left endpoint of a saturated region.  There exists $\ep>0$ such that, for $|z-\be_j|<\ep$,
\begin{equation}\label{mr13}
\begin{aligned}
P_N(z)&=e^{\frac{N}{2}(V(z)+l)}\bigg\{N^{1/6}\psi_{\be_j}(z)^{1/4}\mathcal{B}_1(z)\left[-e^{\pm  \frac{i\Om_{j,N}}{2}}\mcal_1(z)+e^{\mp  \frac{i\Om_{j,N}}{2}}\mcal_2(z)+O(N^{-1})\right] \\
&\quad - N^{-1/6}\psi_{\be_j}(z)^{-1/4}\mathcal{B}_2(z)\left[e^{\pm  \frac{i\Om_{j,N}}{2}}\mcal_1(z)+e^{\mp  \frac{i\Om_{j,N}}{2}}\mcal_2(z)+O(N^{-1})\right]\bigg\}
\end{aligned}
\end{equation}
for $\pm \Im z >0$, where
\begin{equation}\label{mr14}
\begin{aligned}
\mathcal{B}_1(z)&=\cos(N\pi z)\Ai(N^{2/3}\psi_{\be_j}(z))+\sin(N\pi z)\Bi(N^{2/3}\psi_{\be_j}(z)), \\
\mathcal{B}_2(z)&=\cos(N\pi z)\Ai '(N^{2/3}\psi_{\be_j}(z))+\sin(N\pi z)\Bi '(N^{2/3}\psi_{\be_j}(z)).
\end{aligned}
\end{equation}

The point $\al_{j+1}$ is the left endpoint of a band and the right endpoint of a void.  There exists $\ep>0$ such that, for $|z -\al_{j+1}|<\ep$,
\begin{equation}\label{mr15}
\begin{aligned}
P_N(z)&=e^{\frac{N}{2}(V(z)+l)}\left\{N^{1/6}\psi_{\al_{j+1}j}(z)^{1/4}\mathcal{B}_3(z)\left[e^{\pm  \frac{i\Om_{j,N}}{2}}\mcal_1(z)+e^{\mp  \frac{i\Om_{j,N}}{2}}\mcal_2(z)+O(N^{-1})\right]\right. \\
&\quad \left. - N^{-1/6}\psi_{\al_{j+1}}(z)^{-1/4}\mathcal{B}_4(z)\left[e^{\pm  \frac{i\Om_{j,N}}{2}}\mcal_1(z)-e^{\mp  \frac{i\Om_{j,N}}{2}}\mcal_2(z)+O(N^{-1})\right]\right\}
\end{aligned}
\end{equation}
for $\pm \Im z >0$, where
\begin{equation}\label{mr16}
\begin{aligned}
\mathcal{B}_3(z)&=\cos(N\pi z)\Ai(N^{2/3}\psi_{\al_{j+1}}(z))-\sin(N\pi z)\Bi(N^{2/3}\psi_{\al_{j+1}}(z)), \\
\mathcal{B}_4(z)&=\cos(N\pi z)\Ai '(N^{2/3}\psi_{\al_{j+1}}(z))-\sin(N\pi z)\Bi '(N^{2/3}\psi_{\al_{j+1}}(z)).
\end{aligned}
\end{equation}
\end{theo}

\begin{rem}
Although the above theorems are presented for real analytic potential $V(x),$ these results may be extended to potentials which are continuous and piecewise real analytic, assuming that the points of non-analyticity lie strictly within saturated regions and voids.  In this case the preceding results hold, and the asymptotic solution to the associated Riemann-Hilbert Problem does not require local analysis near the points of non-analyticity (see \cite{BL2}).
\end{rem}

Before continuing with the proofs of these theorems, we would also like to remark that the results obtained in this paper match the results obtained in \cite{BKMM} for polynomials orthogonal on a lattice which sits inside a finite interval.  Consequently, many corollaries discussed in \cite{BKMM} also apply to infinite lattices.  In particular, the authors of \cite{BKMM} discuss the particle statistics of the  {\it discrete orthogonal polynomial ensemble} in different regions of a finite interval of the real line, which are based on asymptotic properties of the associated orthogonal polynomials.  The results of this paper imply that their results may be extended to discrete orthogonal polynomial ensembles on an infinite (regular) lattice.  Of particular interest may be the {\it discrete sine kernel} as the scaling limit of the reproducing kernel in the interior of bands, the {\it Airy kernel} as the scaling limit of the reproducing kernel at band end-points, the {\it Tracy-Widom distribution} for the location of the left- and right-most particle, and exponential estimates for all correlation functions in voids and saturated regions.

The rest of the paper is organized as follows.  In Section \ref{IP}, we reformulate the orthogonal polynomials (\ref{in3}) as the solution to an interpolation problem of complex analysis.  In Section \ref{red}, we reduce the interpolation problem to a Riemann-Hilbert Problem which can be solved by steepest descent analysis, which is done in Sections \ref{ft}-\ref{tt}.  Finally, in Section \ref{proof}, we give proofs of the preceding theorems.

%\section{Universality of the scaling limits in the discrete orthogonal polynomial ensemble}

%Consider the collection $\La_N$ of finite subsets $\la=\{\la_1,\ldots,\la_N\}\subset L_N$ of the  lattice
%$L_N$. Observe that the number $N$ of points in $\la$ is equal to the index of $L_N$.
%The {\it discrete orthogonal polynomial ensemble} is a probability distribution on $\La_N$, 
%given by the formula
%\begin{equation}\label{ope1}
%\mu_N(\la)=\frac{1}{Z_N}\,\prod_{j>k}(\la_j-\la_k)^2
%\prod_{j=1}^N e^{-NV(\la_j)},
%\end{equation}
%where $Z_N$ is the partition function,
%\begin{equation}\label{ope2}
%Z_N=\sum_{\la\in\La_N}\prod_{j>k}(\la_j-\la_k)^2
%\prod_{j=1}^N e^{-NV(\la_j)}.
%\end{equation}
%The discrete orthogonal polynomial ensemble is determinantal, and its  
%$m$-point correlation function is equal to the determinant,
%\begin{equation}\label{ope3}
%R_{mN}(x_1,\dots,x_m)=\det \left( K_N(x_k,x_l)\right)_{k,l=1}^m,
%\end{equation} 
%where
%\begin{equation}\label{ope4}
%K_N(x,y)=\sum_{n=0}^{N-1}\psi_n(x)\psi_n(y)
%\end{equation} 
%and
%\begin{equation}\label{ope5}
%\psi_n(x)=\frac{1}{\sqrt{h_n}}P_n(x)e^{-NV(x)/2}.
%\end{equation} 
%see, e.g., \cite{Meh}, \cite{Sos}. Asymptotic formulae for the orthogonal polynomials allow us to obtain
%the scaling limits of the correlation functions $R_{mN}$ as $N\to\infty$.

%By the Christoffel-Darboux formula,
%\begin{equation}\label{ope6}
%K_N(x,y)=\ga_N
%\frac{\psi_N(x)\psi_{N-1}(y)-\psi_{N-1}(x)\psi_N(y)}{x-y}\,.
%\end{equation} 

\section{Interpolation Problem}\label{IP}
We will evaluate the asymptotics of the discrete orthogonal polynomials described above via a steepest descent asymptotic analysis of a Riemann-Hilbert problem.  To that end, we first introduce the following interpolation problem.  

{\bf Interpolation Problem}. For a given $N=0,1,\ldots$, find a $2\times 2$ matrix-valued function
$\bold P_N(z)=(\bold P_N(z)_{ij})_{1\le i,j\le 2}$ with the following properties:
\begin{enumerate}
\item
{\it Analyticity}: $\bold P_N(z)$ is an analytic function of $z$ for $z\in\C\setminus L_N$.
\item
{\it Residues at poles}: At each node $x\in L_N$, the elements $\bold P_N(z)_{11}$ and
$\bold P_N(z)_{21}$ of the matrix $\bold P_N(z)$ are analytic functions of $z$, and the elements $\bold P_N(z)_{12}$ and
$\bold P_N(z)_{22}$ have a simple pole with the residues,
\begin{equation} \label{IP1}
\underset{z=x}{\rm Res}\; \bold P_N(z)_{j2}=w_N(x)\bold P_N(x)_{j1},\quad j=1,2.
\end{equation}
\item
{\it Asymptotics at infinity}: There exists a function $r(x)>0$ on  $L_N$ such that 
\begin{equation} \label{IP2a}
\lim_{x\to\infty} r(x)=0,
\end{equation} 
and such that as $z\to\infty$, $\bold P_N(z)$ admits the asymptotic expansion,
\begin{equation} \label{IP2}
\bold P_N(z)\sim \left( \I+\frac {\bold P_1}{z}+\frac {\bold P_2}{z^2}+\ldots\right)
\begin{pmatrix}
z^N & 0 \\
0 & z^{-N}
\end{pmatrix},\qquad z\in \C\setminus \left[\bigcup_{x\in L_N}^\infty D\big(x,r(x)\big)\right],
\end{equation}
where $D(x,r(x))$ denotes a disk of radius $r(x)>0$ centered at $x$ and $\I$ is the identity matrix.
\end{enumerate}

It is not difficult to see (see \cite{BoB} and \cite{BKMM}) that 
the IP has a unique solution, which is
\begin{equation} \label{IP3}
\bold P_N(z)=
\begin{pmatrix}
P_N(z) & C(w_N P_N)(z) \\
(h_{N-1})^{-1}P_{N-1}(z) & (h_{N-1})^{-1}C(w_NP_{N-1})(z)
\end{pmatrix},
\end{equation}
where the Cauchy transformation $C$ is defined by the formula,
\begin{equation} \label{IP4}
C(f)(z)=\sum_{x\in L_N}\frac{f(x)}{z-x}\,.
\end{equation}
Because of the orthogonality condition, as $z\to\infty$,
\begin{equation} \label{IP5}
C(w_NP_n)(z)=\sum_{x\in L_N}\frac{w_N(x)P_n(x)}{z-x}
\sim \sum_{x\in L_N} w_N(x)P_n(x)\sum_{j=0}^\infty \frac{x^j}{z^{j+1}}=\frac{h_n}{z^{n+1}}
+\sum_{j=n+2}^\infty \frac{a_j}{z^j}\,,
\end{equation}
which justifies asymptotic expansion (\ref{IP2}), and have that 
\begin{equation} \label{IP6}
h_N=[\bold P_1]_{12},\qquad h_{N-1}^{-1}=[\bold P_1]_{21}.
\end{equation}
Furthermore, the recurrence coefficients in equation (\ref{in5}) are given by
\begin{equation}\label{IP7}
\ga_N^2=[\bold P_1]_{12}[\bold P_1]_{21} \quad ; \quad \be_{N-1}=\frac{[\bold P_2]_{21}}{[\bold P_1]_{21}}-[\bold P_1]_{11}.
\end{equation}

\section{Reduction of IP to RHP}\label{red}
We would like to reduce the Interpolation Problem to a Riemann-Hilbert Problem (RHP).  Introduce the function
\begin{equation}\label{red1}
\Pi(z)=\frac{\sin{N\pi z}}{N\pi}.
\end{equation}
Notice that
\begin{equation}\label{red2}
\Pi(x_k)=0, \quad \Pi'(x_k)=\exp\left(iN\pi x_k\right)=(-1)^k, \quad \textrm{for} \quad x_k=\frac{k}{N} \in L_N.
\end{equation}

Introduce the upper triangular matrices,
\begin{equation} \label{redp4}
\bold D^u_{\pm}(z)
=\begin{pmatrix}
1 & -\frac{w_N(z)}{\Pi(z)}e^{\pm iN \pi z}   \\
0 & 1
\end{pmatrix},
\end{equation}
and the lower triangular matrices,
\begin{equation} \label{redp5}
\bold D^l_{\pm}
=\begin{pmatrix}
\Pi(z)^{-1} & 0   \\
-\frac{1}{w_N(z)}e^{\pm iN\pi z} & \Pi(z)
\end{pmatrix}
=\begin{pmatrix}
\Pi(z)^{-1} & 0   \\
0 & \Pi(z)
\end{pmatrix}
\begin{pmatrix}
1 & 0   \\
-\frac{1}{\Pi(z)w_N(z)}e^{\pm iN\pi z} & 1
\end{pmatrix}.
\end{equation}
Define the matrix-valued functions,
\begin{equation} \label{redp6}
\bold R^u_N=\bold P_N(z)\times
\left\{
\begin{aligned}
& \bold D^u_+(z) \quad \textrm{when}\quad \Im z\ge0 \\
& \bold D^u_-(z) \quad \textrm{when}\quad \Im z\le0,
\end{aligned}
\right.
\end{equation}
and 
\begin{equation} \label{redp7}
\bold R^l_N=\bold P_N(z)\times
\left\{
\begin{aligned}
& \bold D^l_+(z),\quad \textrm{when}\quad \Im z\ge0 \\
& \bold D^l_-(z),\quad \textrm{when}\quad \Im z\le0.
\end{aligned}
\right.
\end{equation}
From (\ref{IP3}) we have that
\begin{equation} \label{redp8}
\begin{aligned}
\bold R_N^u(z)&=
\begin{pmatrix}
P_N(z) & -\frac{w_N(z)P_N(z)}{\Pi(z)}e^{\pm iN\pi z}
+C(w_N P_N)(z)   \\
h_{N-1}^{-1}P_{N-1}(z)  
& -\frac{w_N(z)h_{N-1}^{-1}P_{N-1}(z)}{\Pi(z)}e^{\pm iN\pi z}
+h_{N-1}^{-1}C(w_N P_{N-1})(z) 
\end{pmatrix} \\
&\quad\textrm{when}\quad \pm \Im z\ge 0,
\end{aligned}
\end{equation}
and
\begin{equation} \label{redp9}
\begin{aligned}
\bold R_N^l(z)&=
\begin{pmatrix}
\frac{P_N(z)}{\Pi(z)}-\frac{C(w_N P_N)(z)}{w_N(z)}e^{\pm iN\pi z} 
  & \Pi(z)C(w_N P_N)(z)   \\
\frac{h_{N-1}^{-1}P_{N-1}(z)}{\Pi(z)}-\frac{h_{N-1}^{-1}C(w_N P_{N-1})(z)}{w_N(z)}e^{\pm iN\pi z} 
  & \Pi(z)h_{N-1}^{-1}C(w_N P_{N-1})(z)  
\end{pmatrix} \\
&\quad\textrm{when}\quad \pm \Im z\ge 0.
\end{aligned}
\end{equation}
Observe that the functions $\bold R_N^u(z)$, $\bold R_N^l(z)$ are meromorphic on the closed upper and lower complex planes and they are two-valued on the real axis. Their possible poles
are located on the lattice $L_N$. An important result is that, in fact, 
due to some cancellations they do not have any poles at all.
We have the following proposition. 

\begin{prop}
The matrix-valued functions $\bold R_N^u(z)$ and $\bold R_N^l(z)$ have no poles and
on the real line they satisfy the following jump conditions at $x\in\R$:
\begin{equation} \label{redp9a}
\bold R_{N+}^u(x)=\bold R_{N-}^u(x) j_R^u(x),\qquad
j_R^u(x)=
\begin{pmatrix}
1 & -2N\pi iw_N(x)\, \\
0 & 1
\end{pmatrix},
\end{equation}
and
\begin{equation} \label{redp9b}
\bold R_{N+}^l(x)=\bold R_{N-}^l(x) j_R^l(x),\qquad
j_R^l(x)=
\begin{pmatrix}
1 & 0 \\
-\frac{2N\pi i}{w_N(x)} & 1
\end{pmatrix},
\end{equation}
\end{prop}

\begin{proof} It follows from the definition of $\bold R_N^u(z)$ 
that all possible poles of $\bold R_N^u(z)$ are located on the lattice $L_N$.
Let us show that the residue of all these poles is equal to zero.
Consider any $x_k\in L_N$. The residue of the matrix element $\bold R_{N,12}^u(z)$ at $x_k$
is  equal to
\begin{equation} \label{redp10}
\underset{z=x_k}{\rm Res}\;\bold R_{N,12}^u(z)=
-\frac{w_N(x_k)P_N(x_k)}{(-1)^k}(-1)^k+w_N(x_k)P_N(x_k)=0.
\end{equation}
Similarly we get that
\begin{equation} \label{redp11}
\underset{z=x_k}{\rm Res}\;\bold R_{N,22}(z)=0,
\end{equation}
hence $\bold R_N^u(z)$ has no pole at $x_k$.

Similarly, the residue of the matrix element $\bold R_{N,11}^l(z)$ at $x_k$ is  equal to
\begin{equation} \label{redp12}
\underset{z=x_k}{\rm Res}\;\bold R_{N,11}^l(z)=\frac{P_N(x_k)}{(-1)^k}
-\frac{w_N(x_k)P_N(x_k)(-1)^k}{w_N(x_k)}=0.
\end{equation}
In the same way we obtain that
\begin{equation} \label{redp13}
\underset{z=x_k}{\rm Res}\;\bold R_{N,21}(z)=0.
\end{equation}
In the entry $\bold R_{N,21}^l(z)$, the pole of the function $C(w_N P_N)(z)$ at $z=x_k$ 
is cancelled by the zero of the function $\Pi(z)$,
hence $\bold R_{N,21}^l(z)$ has no pole at $x_k$. Similarly, $\bold R_{N,22}^l(z)$ has no pole at $x_k$
as well, hence $\bold R_N^l(z)$ has no pole at $x_k$. 

Let us evaluate the jump matrices for $x\in\R$. From (\ref{redp6}) we have that
\begin{equation} \label{redp14}
j_R^u(x)=\bold D^u_-(x)^{-1}\bold D^u_+(x)=
\begin{pmatrix}
1 & -\frac{w_N(x)}{\Pi(x)} 2i\sin {N\pi x}   \\
0 & 1
\end{pmatrix}
=\begin{pmatrix}
1 & -2N\pi iw_N(x)\, \\
0 & 1
\end{pmatrix},
\end{equation}
which proves (\ref{redp9a}). Similarly,
\begin{equation} \label{redp15}
j_R^l(x)=\bold D^l_-(x)^{-1}\bold D^l_+(x)=
\begin{pmatrix}
1 &  0  \\
-\frac{1}{\Pi(x)w_N(x)} 2i\sin {N\pi x} & 1
\end{pmatrix}
=\begin{pmatrix}
1 & 0 \\
-\frac{2N\pi i}{w_N(x)} & 1
\end{pmatrix},
\end{equation}
which proves (\ref{redp9b}).\end{proof}

To reduce the Interpolation Problem to a Riemann-Hilbert Problem, we follow the work \cite{BKMM} with some modifications.  Consider the oriented contour $\Sigma$ on the complex plane depicted in Fig. \ref{fig2}, in which the horizontal lines are $\Im z=\ep, 0,-\ep$, where $\ep >0$ is a small positive constant 
which will be determined later, and the vertical segments pass through
the endpoints of saturated intervals. Consider the regions 

\begin{figure}
\vskip 1cm
\scalebox{0.8}{\includegraphics{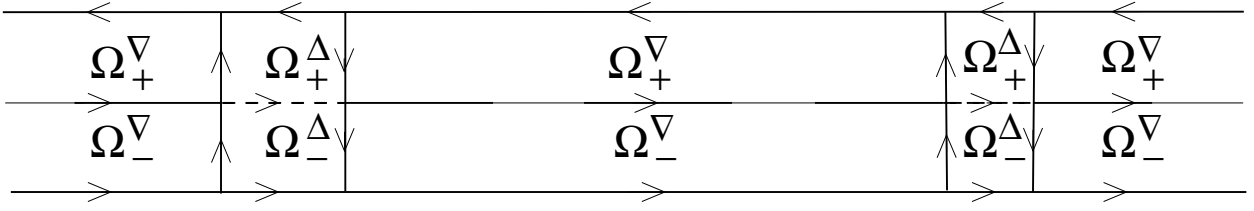}}
\caption{The contour $\Sigma$ arising from the hypothetical equilibrium measure in Figure \ref{fig1}, dividing an $\ep$-neighborhood of the real line into the regions $\Om_\pm^\Delta$ and $\Om_\pm^\nabla$. }
\label{fig2}
\end{figure}

\begin{equation}\label{red21}
\begin{aligned}
\Om^\nabla_{\pm}&=\{I^0 \cup I^-\} \times (0,\pm i \ep) \\
\Om^\De_{\pm}&=I^+ \times (0,\pm i \ep)
\end{aligned}
\end{equation}
bounded by the contour $\Sigma$.  
Define 
\begin{equation} \label{red22}
\bold R_N(z)=
\left\{
\begin{aligned}
&\bold K_N\bold R_N^u(z)\bold K_N^{-1},\quad \textrm{for}\quad z\in\Om_{\pm}^\nabla,\\
&\bold K_N\bold R_N^l(z)\bold K_N^{-1},\quad \textrm{for}\quad z\in\Om_{\pm}^\De,\\
&\bold K_N \bold P_N(z)\bold K_N^{-1}, \quad \textrm{otherwise}.
\end{aligned}
\right.
\end{equation}
where $\bold K_N=\begin{pmatrix} 1 & 0 \\ 0 & -2iN\pi \end{pmatrix}$.

If $A\subset \C$ is a set on the complex plane and $b\in\C$
then, as usual, we denote
\begin{equation} \label{red3a}
A+b=\{z=a+b,\; a\in A\}.
\end{equation}

\begin{prop}
The matrix-valued function $\bold R_N(z)$ has the following jumps on the contour $\Sigma$:
\begin{equation} \label{red23}
\bold R_{N+}(z)=\bold R_{N-}(z)j_R(z),
\end{equation}
where
\begin{equation} \label{red24}
j_R(z)=
\left\{
\begin{aligned}
& \begin{pmatrix} 1 & w_N(z) \\ 0 & 1 \end{pmatrix} \quad \textrm{when}\quad z\in I^- \cup I^0 \\
& \begin{pmatrix} 1 & 0 \\ -(2N\pi)^2w_N(z)^{-1} & 1 \end{pmatrix} \quad \textrm{when}\quad z\in  I^+,\\
& \bold K_N\bold D_{\pm}^u(z)\bold K_N^{-1}=\begin{pmatrix} 1 & \frac{1}{2iN\pi}
\frac{w_N(z)e^{\pm iN \pi z}}{\Pi(z)} \\ 0 & 1 \end{pmatrix} \quad \textrm{when}\quad z\in  \{I^- \cup I^0\} \pm i\ep,\\
&\bold K_N\bold D_{\pm}^l(z)\bold K_N^{-1}=
\begin{pmatrix}
\Pi(z)^{-1} & 0 \\ 
2iN\pi \frac{e^{\pm iN\pi z}}{w_N(z)} & \Pi(z) 
\end{pmatrix}
\quad \textrm{when }\quad z\in I^+ \pm i\ep\\
& \bold K_N\bold D_{\pm}^l(z)^{-1}\bold D_{\pm}^u(z)\bold K_N^{-1}=
\begin{pmatrix} 
\Pi(z) & \frac{1}{2N\pi i}w_N(z)e^{\pm iN\pi z} \\ 
-2N\pi iw_N(z)^{-1}e^{\pm iN\pi z} 
& \mp 2N\pi i e^{\pm iN\pi z} 
\end{pmatrix} \\
&\hspace{2.5 cm} \textrm{when}\quad z\in  (0,\pm i\ep)+\be_j \quad
\textrm{or}\quad z\in  (0,\pm i\ep)+ \al_{j+1} \quad \textrm{for} \quad j\in \acal_s.
\end{aligned}
\right.
\end{equation}
\end{prop}

\section{First transformation of the RHP}\label{ft}

Define the matrix function $\bold T_N(z)$ as follows from the equation,
\begin{equation}\label{ft1}
\bold R_N(z)=e^{\frac{Nl}{2}\sigma_3}\bold T_N(z)e^{N(g(z)-\frac{l}{2})\sigma_3},
\end{equation}
where $l$ is the Lagrange multiplier, the function $g(z)$ is described 
in Section \ref{equilibrium}, and $\sigma_3=\begin{pmatrix} 1 & 0 \\ 0 &-1 \end{pmatrix}$ is the third Pauli matrix.  
Then $\bold T_N(z)$ satisfies the following Riemann-Hilbert Problem:
\begin{enumerate}
  \item 
  $\bold T_N(z)$ is analytic in $\C \setminus \Sg$. 
  \item $\bold T_{N+}(z)=\bold T_{N-}(z)j_T(z)$ for $z\in\Sg$, where
  \begin{equation}\label{ft2}
   j_T(z)=
  \left\{
  \begin{aligned}
 & e^{N(g_-(z)-\frac{l}{2})\sigma_3}j_R(z)e^{-N(g_+(z)-\frac{l}{2})\sigma_3} \quad \textrm{for} \quad z \in \R \\
  &e^{N(g(z)-\frac{l}{2})\sigma_3}j_R(z)e^{-N(g(z)-\frac{l}{2})\sigma_3} \quad \textrm{for} \quad z \in \Sg \setminus \R.
  \end{aligned}\right.
  \end{equation}
  \item As $z \to \infty,$
   \begin{equation} \label{ft3}
\bold T_N(z)\sim  \I+\frac {\bold T_1}{z}+\frac {\bold T_2}{z^2}+\ldots.
\end{equation}
\end{enumerate}
%From (\ref{g2}) we have that
%\begin{equation} \label{ft3a}
%g(z)=\log z+O(z^{-1}),
%\qquad \textrm{as} \quad z \to \infty.
%\end{equation}
%This implies that
%\begin{equation} \label{ft3b}
%\bold T_1=e^{-Nl}\bold R_1.
%\end{equation}
Let's take a closer look at the behavior of the jump matrix $j_T$ described in (\ref{ft2}) on the horizontal segments of $\Sg$.  We have that 
\begin{equation} \label{ft4}
j_T(z)=
\left\{
\begin{aligned}
&\begin{pmatrix} e^{-NG(z)} & w_N(z)e^{N(g_+(z)+g_-(z)-l)} \\ 0 & e^{NG(z)} \end{pmatrix}
\quad \textrm{when}\quad z\in I^0 \cup I^- \\
&\begin{pmatrix} e^{-NG(z)} & 0 \\ -(2N\pi)^2e^{-N(g_+(z)+g_-(z)-V(z)-l)} & e^{NG(z)} \end{pmatrix}\quad \textrm{when}\quad z\in  I^+ \\
&\begin{pmatrix} 1 & \pm\frac{e^{N(2g(z)-l-V(z))}}{1-e^{\mp 2iN\pi x}e^{\ep 2N\pi}} \\ 0 & 1 \end{pmatrix}
\quad \textrm{when }\quad z=x\pm i\ep \in \{ I^-\pm i\ep\} \\
&\begin{pmatrix} 1 & \pm\frac{e^{\pm NG(z)}}{1-e^{\mp 2iN\pi x}e^{\ep 2N\pi}} \\ 0 & 1 \end{pmatrix}\quad \textrm{when }\quad z=x\pm i\ep \in \{ I^0 \pm i\ep \} \\
&\begin{pmatrix} \Pi(z)^{-1} & 0 \\ 2iN\pi e^{\pm iN\pi x}e^{-N(2g(z)-l-V(z))} & \Pi(z)\end{pmatrix}\quad \textrm{when}\quad z=x \pm i \ep \in  \{I^+ \pm i\ep\}.
\end{aligned}\right.
\end{equation}

\section{Second transformation of the RHP}

The second transformation is based on two observations.  The first is the well known 
``opening of the lenses'' in a neighborhood of the unconstrained support of the equilibrium measure.  
Namely, notice that, for $x\in I^0$, the jump matrix $j_T(x)$ factorizes as 
\begin{equation}\label{st1}
\begin{aligned}
j_T(x)=\begin{pmatrix} e^{-NG(x)} & 1 \\ 0 & e^{NG(x)} \end{pmatrix}&=\begin{pmatrix} 1 & 0 \\ e^{NG(x)} & 1 \end{pmatrix}
\begin{pmatrix} 0 & 1 \\ -1 & 0 \end{pmatrix} \begin{pmatrix} 1 & 0 \\ e^{-NG(x)} & 1 \end{pmatrix} \\
&=j_-(x)j_M j_+(x),
\end{aligned}
\end{equation}
which allows us to reduce the jump matrix $j_T$ to the one $j_M$ plus asymptotically small 
jumps on the lens boundaries.
The second observation consists of two facts.  Firstly, the jumps on the segments 
$ I^+ \pm i\ep$ behave, for large $N$, as $\pm e^{\pm iN\pi z}$.  
Secondly, note that, for $x\in I^+$, $G(x)$ is a linear function with slope $-2\pi i$.  
With these facts in mind, we make the second transformation of the RHP.  Let
\begin{equation}\label{st2}
\begin{aligned}
&\bold S_N(z)=\left\{
\begin{aligned}
&\bold T_N(z)j_+(z)^{-1} \quad \textrm{for} \quad z\in I^0 \times (0,i\ep) \\
&\bold T_N(z)j_-(z)\quad \textrm{for} \quad z\in I^0 \times (0,-i\ep) \\
&\bold T_N(z)\bold A_+(z)\quad \textrm{for} \quad z\in I^+ \times (0,i\ep) \\
&\bold T_N(z)\bold A_-(z)\quad \textrm{for} \quad z\in I^+ \times (0,-i\ep) \\
&\bold T_N(z)  \quad \textrm{otherwise}.
\end{aligned}\right. \\ 
\textrm{where} \ &\bold A_+(z) =\begin{pmatrix} -\frac{1}{2N\pi i}e^{- iN\pi z} & 0 \\ 0 & -2N\pi i e^{iN\pi z} \end{pmatrix} \ \textrm{and} \ \bold A_-(z) =\begin{pmatrix} \frac{1}{2N\pi i}e^{iN\pi z} & 0 \\ 0 & 2N\pi i e^{-iN\pi z} \end{pmatrix}.
\end{aligned}
\end{equation}

This function satisfies a similar RHP to $\bold T$, but jumps now occur on a new contour, $\Sigma_S$, 
which is obtained from $\Sg$ by adding the segments $(\al_1-i\ep, \al_1+i\ep)$, $(\be_q-i\ep, \be_q+i\ep)$,$(\al_{j+1}-i\ep, \al_{j+1}+i\ep)$, $(\be_j-i\ep, \be_j+i\ep)$ for $j \in \acal_v$, see Fig. \ref{sigma_S}.

%%%%%%%%%%%%% Fig.  %%%%%%%%%%%%%%
\begin{center}
 \begin{figure}[h]
\begin{center}
   \scalebox{0.8}{\includegraphics{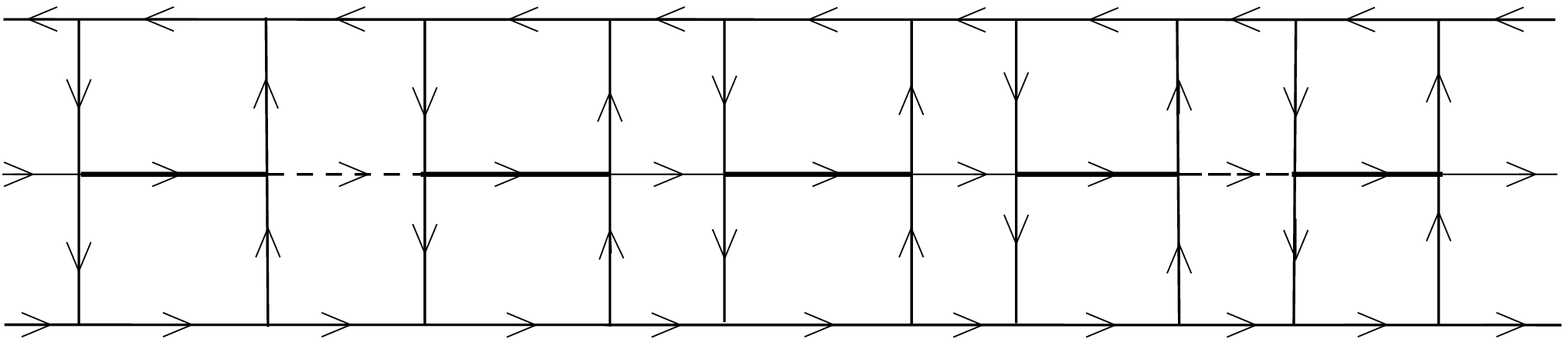}}
\end{center}
        \caption[The contour $\Sigma_S$.]{The contour $\Sigma_S$ arising from the hypothetical equilibrium measure shown in Figure \ref{fig1}.}
\label{sigma_S}
    \end{figure}
\end{center}
%%%%%%%%%%%%%%%%%%%%%%%%%%%%%%%%%%%

On horizontal segments, we have that
\begin{equation}\label{st3}
j_S(z)=\left\{
\begin{aligned}
& \begin{pmatrix} 0 & 1 \\ -1 & 0 \end{pmatrix}
 \quad &\textrm{for} \quad z\in I^0 \\
& \begin{pmatrix} 1+O(e^{-2\ep N\pi}) & O(e^{N(G(z)-2\ep \pi)})  \\ -e^{-NG(z)} & 1\end{pmatrix} \quad &\textrm{for} \quad z-i\ep\in I^0 \\
&\begin{pmatrix} 1+O(e^{-2\ep N\pi}) & O(e^{N(-G(z)-2\ep \pi)}) \\ e^{NG(z)} & 1\end{pmatrix} \quad &\textrm{for} \quad z+i\ep\in I^0 \\
&\begin{pmatrix} 1+O(e^{-2\ep N\pi}) & 0 \\ 2iN\pi e^{-N(2g(z)-l-V(z))} & 1+O(e^{-2\ep N\pi}) \end{pmatrix} \quad &\textrm{for} \quad z \in \{I^+\pm i\ep\} \\
&\begin{pmatrix} -e^{-iN\Omega_j} & 0 \\ -e^{-N(g_+(z)+g_-(z)-l-V(z))} & -e^{iN\Omega_j} \end{pmatrix} \quad &\textrm{for}\quad z \in (\be_j,\al_{j+1}), \quad j \in \mathcal{A}_s \\
& \begin{pmatrix} 1 & {e^{N(2g(z)-l-V(z))}}O(e^{-2\ep N\pi}) \\ 0 & 1 \end{pmatrix}
\quad &\textrm{for}\quad z=x\pm i\ep \in  \{I^-\pm i\ep\} \\
&\begin{pmatrix} e^{-iN\Om_j} & e^{N(g_+(z)+g_-(z)-l-V(z))} \\ 0 & e^{iN\Om_j} \end{pmatrix}
\quad &\textrm{for}\quad z \in (\be_j,\al_{j+1}), \quad j \in \mathcal{A}_v. 
\end{aligned}\right.
\end{equation}
By formula (\ref{g5}) for the $G$-function and the upper constraint on the density $\rho$, we obtain that, for sufficiently small $\ep > 0$ and $x \in (\al_j,\be_j)$, 
\begin{equation}\label{st4}
0< \mp \Re G(x\pm i\ep) = 2\pi \rho(x) +O(\ep^2) < 2\pi \ep +  O(\ep^2).
\end{equation}
This, combined with property (\ref{g7}) of the $g$-function, implies that all jumps on horizontal segments are exponentially close to the identity matrix, provided that they are bounded away from the segment $(\al_1,\be_q)$.

\section{Model RHP}

The model RHP appears when we drop in the jump matrix $j_S(z)$ the terms that vanish as $N\to\infty$:
\begin{enumerate}
  \item 
  $\bold M(z)$ is analytic in $\C \setminus [\al_1,\be_q]$.
  \item $\bold M_{+}(z)=\bold M_{-}(z)j_M(z)$ for $z\in[\al_1,\be_q]$, where
  \begin{equation}\label{m1}
   j_M(z)=
  \left\{
  \begin{aligned}
 & \begin{pmatrix} 0 & 1 \\ -1 & 0 \end{pmatrix} \quad \textrm{for} \quad z \in I^0 \\
&\; e^{-i\Om_{j,N}\sigma_3} \quad \textrm{for} \quad z \in (\be_j,\al_{j+1}),
  \end{aligned}\right.
  \end{equation}
 \item As $z \to \infty$,
\begin{equation} \label{m2}
\bold M(z)\sim \I+\frac {\bold M_1}{z}+\frac {\bold M_2}{z^2}+\ldots.
\end{equation}
\end{enumerate}

This model problem was first solved in the general multi-cut case in \cite{DKMVZ}, and the solution is given as follows, using the notation introduced in section \ref{results}.
\begin{equation}\label{m20}
\bold M(z) = \bold F (\infty)^{-1} \begin{pmatrix} \frac{\ga(z)+\ga^{-1}(z)}{2}\frac{\th(u(z)+\frac{\Om_N}{2\pi}+d)}{\th(u(z)+d)} &  \frac{\ga(z)-\ga^{-1}(z)}{-2i}\frac{\th(u(z)-\frac{\Om_N}{2\pi}-d)}{\th(u(z)-d)} \\ \frac{\ga(z)-\ga^{-1}(z)}{2i}\frac{\th(u(z)+\frac{\Om_N}{2\pi}-d)}{\th(u(z)-d)} & \frac{\ga(z)+\ga^{-1}(z)}{2}\frac{\th(u(z)-\frac{\Om_N}{2\pi}+d)}{\th(u(z)+d)} \end{pmatrix}
\end{equation}
where
\begin{equation}\label{m21}
\bold F(\infty)=\begin{pmatrix} \frac{\th(u(\infty)+\frac{\Om_N}{2\pi}+d)}{\th(u(\infty)+d)} & 0 \\ 0 & \frac{\th(u(\infty)-\frac{\Om_N}{2\pi}+d)}{\th(u(\infty)+d)} \end{pmatrix}.
\end{equation}
 The asymptotics at infinity are given as
\begin{equation}\label{m28}
\bold M(z)=\bold I + \frac{\bold M_1}{z} + O(z^{-2}).
\end{equation}
%where
%\begin{equation}\label{m29}
%\bold M_1 = \begin{pmatrix} 0 & \frac{\th(u(\infty)-\frac{\Om_N}{2\pi}-d)\th(u(\infty)+d)}{\th(u(\infty)+\frac{\Om_N}{2\pi}+d)\th(u(\infty)-d)}\sum_{j=1}^q \frac{\be_j-\al_j}{-4i} \\ \frac{\th(u(\infty)+\frac{\Om_N}{2\pi}-d)\th(u(\infty)+d)}{\th(u(\infty)-\frac{\Om_N}{2\pi}+d)\th(u(\infty)-d)}\sum_{j=1}^q \frac{\be_j-\al_j}{4i} & 0 \end{pmatrix}.
%\end{equation}

\section{Parametrix at band-void edge points}

We now consider small disks $D(\al_j,\ep)$ for $j-1 \in \acal_v \cup \{0\}$, and $D(\be_j,\ep)$ for $j \in \acal_v \cup \{q\}$, centered at the endpoints of bands which are adjacent to a void.  Denote 
\begin{equation}\label{pm0a}
D=\left(\underset{j-1 \in \acal_v \cup \{0\}}{\bigcup} D(\al_j,\ep)\right) \bigcup \left(\underset{j \in \acal_v \cup \{q\}}{\bigcup} D(\be_j,\ep)\right).
\end{equation}
We will seek a local parametrix $\bold U_N(z)$ defined on $D$ such that\begin{enumerate}
\item 
\begin{equation}\label{pm0}
\bold U_N(z) \ \textrm{is analytic on} \ D \setminus \Sigma_S.
\end{equation}
\item
\begin{equation}\label{pm1}
\bold U_{N+}(z)=\bold U_{N-}(z)j_S(z) \quad \textrm{for} \quad z\in D \cap \Sigma_S.
\end{equation}
\item
\begin{equation}\label{pm2}
\bold U_N(z)=\bold M(z) \big(I+O(N^{-1})\big) \quad \textrm{uniformly for} \  z\in \partial D.
\end{equation}
\end{enumerate}
We first construct the parametrix near $\be_j$ for $j \in \acal_v$.  The jumps $j_S$ are given by
\begin{equation}\label{pm3}
j_S(z)=
\left\{
\begin{aligned}
&\begin{pmatrix} 0 & 1 \\ -1 & 0 \end{pmatrix} \quad \textrm{for} \ z\in (\be_j-\ep,\be_j) \\
&\begin{pmatrix} 1& 0 \\- e^{-NG(z)} & 1 \end{pmatrix}\quad \textrm{for} \ z\in (\be_j,\be_j+i\ep) \\
&\begin{pmatrix} 1& 0 \\ e^{NG(z)} & 1 \end{pmatrix}\quad \textrm{for} \ z\in (\be_j,\be_j-i\ep) \\
&\begin{pmatrix} e^{-NG(z)} & e^{N(g_+(z)+g_-(z)-V(z)-l)} \\ 0 &  e^{NG(z)} \end{pmatrix} \quad \textrm{for} \ z\in (\be_j,\be_j+\ep).
\end{aligned}\right.
\end{equation}
If we let
\begin{equation}\label{pm4}
\bold U_N(z)=\bold Q_N(z)e^{-N(g(z)-\frac{V(z)}{2}-\frac{l}{2})\sigma_3},
\end{equation}
then the jump conditions on $\bold Q_N$ become
\begin{equation}\label{pm5}
\bold Q_{N+}(z)=\bold Q_{N-}(z)j_Q(z)
\end{equation}
where
\begin{equation}\label{pm6}
j_Q(z)=
\left\{
\begin{aligned}
&\begin{pmatrix} 0 & 1 \\ -1 & 0 \end{pmatrix} \quad \textrm{for} \ z\in (\be_j-\ep,\be_j) \\
&\begin{pmatrix} 1& 0 \\ -1 & 1 \end{pmatrix}\quad \textrm{for} \ z\in (\be_j,\be_j+i\ep) \\
&\begin{pmatrix} 1& 0 \\ 1 & 1 \end{pmatrix}\quad \textrm{for} \ z\in (\be_j,\be_j-i\ep) \\
&\begin{pmatrix} 1 & 1 \\ 0 &  1 \end{pmatrix} \quad \textrm{for} \ z\in (\be_j,\be_j+\ep).
\end{aligned}\right.
\end{equation}
where orientation is from left to right on horizontal contours, and down to up on vertical contours, according to Figure \ref{sigma_S}.

$\bold Q_N$ can be constructed using Airy functions.  The Airy function solves the differential equation $y''=zy$, and has the following asymptotics at infinity (see, e.g. \cite{Ol}):
\begin{equation}\label{pm6a}
\begin{aligned}
&\Ai (z) = \frac{1}{2\sqrt{\pi}z^{1/4}}e^{-\frac{2}{3}z^{3/2}}\left(1-\frac{5}{48}z^{-3/2}+O(z^{-3})\right) \\
&\Ai' (z) = -\frac{1}{2\sqrt{\pi}}z^{1/4}e^{-\frac{2}{3}z^{3/2}}\left(1+\frac{7}{48}z^{-3/2}+O(z^{-3})\right) \\
\end{aligned}
\end{equation}
as $\ z\to \infty$ with $\arg z \in (-\pi+\ep, \pi-\ep)$ for any $\ep > 0$.
If we let
\begin{equation}\label{pm7}
y_0(z)=\Ai (z), \quad y_1(z)=\omega \Ai (\omega z), \quad y_2(z)=\omega^2 \Ai (\omega^2 z)
\end{equation}
where $\omega = e^{\frac{2\pi i}{3}}$, then the functions $y_0, y_1,$ and $y_2$ satisfy the relation
\begin{equation}\label{pm8}
y_0(z)+y_1(z)+y_2(z)=0.
\end{equation}
If we take
\begin{equation}\label{pm9}
\Phi_{rv}(z)=\left\{
\begin{aligned}
&\begin{pmatrix} y_0(z) & -y_2(z) \\ y_0'(z) & -y_2'(z) \end{pmatrix} \quad \textrm{for} \quad  \arg z \in \left(0,\frac{\pi}{2}\right) \\
&\begin{pmatrix} -y_1(z) & -y_2(z) \\ -y_1'(z) & -y_2'(z) \end{pmatrix} \quad \textrm{for} \quad  \arg z \in \left(\frac{\pi}{2},\pi\right) \\
&\begin{pmatrix} -y_2(z) & y_1(z) \\ -y_2'(z) & y_1'(z) \end{pmatrix} \quad \textrm{for} \quad  \arg z \in \left(-\pi,-\frac{\pi}{2}\right) \\
&\begin{pmatrix} y_0(z) & y_1(z) \\ y_0'(z) & y_1'(z) \end{pmatrix} \quad \textrm{for} \quad  \arg z \in \left(-\frac{\pi}{2},0\right),
\end{aligned}
\right.
\end{equation}
then $\Phi_{rv}$ satisfies jump conditions similar to (\ref{pm6}), but for jumps on rays emanating from the origin rather than from $\be_j$.  We thus need to map the disk $D(\be_j, \ep)$ onto some convex neighborhood of the origin in order to take advantage of the function $\Phi_{rv}$.  Our mapping should match the asymptotics of the Airy function in order to have the matching property (\ref{pm2}).

To this end, notice that, by (\ref{eq23}), for $t \in [\al_j,\be_j]$, as $t\to\be_j$,
\begin{equation}\label{pm9a}
\rho(t)=C(\be_j-t)^{1/2}+O\big((\be-t)^{3/2}\big), \quad C>0.
\end{equation}
It follows that, for $\quad x \in [\al_j,\be_j]\quad \textrm{as} \quad x\to\be_j$,
\begin{equation}\label{pm10}
\int_x^{\be_j}\rho(t)dt=C(\be_j-x)^{3/2}+O\big((\be_j-x)^{5/2}\big) \quad C_0=\frac{2}{3}C.
\end{equation}
Thus,
\begin{equation}\label{pm10a}
\psi_{\be_j}(z)=-\left\{\frac{3\pi}{2}\int_z^{\be_j}\rho(t)dt\right\}^{2/3}
\end{equation}
is analytic at $\be_j$, thus extends to a conformal map from $D(\be_j, \ep)$ (for small enough $\ep$) onto a convex neighborhood of the origin.  Furthermore, 
\begin{equation}\label{pm10b}
\psi_{\be_j}(\be_j)=0 \quad ; \quad \psi_{\be_j}'(\be_j)>0,
\end{equation}
thus $\psi_{\be_j}$ is real negative on $(\be_j-\ep, \be_j)$, and real positive on $(\be_j, \be_j+\ep)$.  Also, we can slightly deform the vertical pieces of the contour $\Sigma_S$ close to $\be_j$, so that
\begin{equation}\label{pm10bb}
\psi_{\be_j}\{D(\be_j,\ep) \cap \Sigma_S\}=(-\ep,\ep) \cup (-i\ep,i\ep).
\end{equation}
We now set
\begin{equation}\label{pm10c}
\bold Q_N(z)=\bold E_N^{\be_j}(z)\Phi_{rv}\big(N^{2/3}\psi_{\be_j}(z)\big)
\end{equation}
so that 
\begin{equation}\label{pm11}
\bold U_N(z)=\bold E_N^{\be_j}(z)\Phi_{rv}\big(N^{2/3}\psi_{\be_j}(z)\big)e^{-N(g(z)-\frac{V(z)}{2}-\frac{l}{2})\sigma_3}
\end{equation}
where
\begin{equation}\label{pm12}
\begin{aligned}
\bold E_N^{\be_j}(z)&=\bold M(z)e^{\pm\frac{i\Om_{j,N}}{2}\sg_3} \bold L_N^{\be_j}(z)^{-1} \quad \textrm{for} \quad \pm \Im z >0, \\
\bold L_N^{\be_j}(z)&=\frac{1}{2\sqrt{\pi}}\begin{pmatrix} N^{-1/6} \psi_{\be_j}^{-1/4}(z) & 0 \\ 0 & N^{1/6} \psi_{\be_j}^{1/4}(z)\end{pmatrix}\begin{pmatrix} 1 & i \\ -1 & i\end{pmatrix},  
\end{aligned}
\end{equation}
and we take the principal branch of $\psi_{\be_j}^{1/4}$, which is positive on $(\be_j, \be_j+\ep)$ and has a cut on  $(\be_j-\ep,\be_j)$. The function $\Phi_{rv}(N^{2/3}\psi_{\be_j}(z))$ has the jumps $j_S$, and we claim that the prefactor $\bold E_N^{\be_j}$ is analytic in $D(\be_j, \ep)$, thus does not change these jumps.  This can be seen, as 
\begin{equation}\label{pmi14}
\bold M_+(z)e^{\frac{i\Omega_{j,N}}{2}\sigma_3}=\bold M_-(z)e^{-\frac{i\Omega_{j,N}}{2}\sigma_3}e^{\frac{i\Omega_{j,N}}{2}\sigma_3}j_Me^{\frac{i\Omega_{j,N}}{2}\sigma_3}
\end{equation}
thus the jump for the function $\bold M(z)e^{\pm\frac{i\Omega_{j,N}}{2}\sigma_3}$ is
\begin{equation}\label{pmi15}
e^{\frac{i\Omega_{j,N}}{2}\sigma_3}j_Me^{\frac{i\Omega_{j,N}}{2}\sigma_3}=\left\{
\begin{aligned}
& e^{\frac{i\Omega_{j,N}}{2}\sigma_3}\begin{pmatrix} 0 & 1 \\ -1 & 0 \end{pmatrix} e^{\frac{i\Omega_{j,N}}{2}\sigma_3}\quad \textrm{for} \quad z \in (\be_j-\ep,\be_j) \\
 &e^{\frac{i\Omega_{j,N}}{2}\sigma_3} e^{-i\Omega_{j,N}\sigma_3} e^{\frac{i\Omega_{j,N}}{2}\sigma_3}\quad \textrm{for} \quad z \in (\be_j,\be_j+\ep)
\end{aligned}\right.
\end{equation}
or equivalently, 
\begin{equation}\label{pmi16}
e^{\frac{i\Omega_{j,N}}{2}\sigma_3}j_Me^{\frac{i\Omega_{j,N}}{2}\sigma_3}=\left\{
\begin{aligned}
&\begin{pmatrix} 0 & 1 \\ -1 & 0 \end{pmatrix} \quad \textrm{for} \quad z \in (\be_j-\ep,\be_j) \\
 &\begin{pmatrix} 1 & 0 \\ 0 &1 \end{pmatrix}\quad \textrm{for} \quad z \in (\be_j,\be_j+\ep),
\end{aligned}\right.
\end{equation}
which is exactly the same as the jump conditions for $\bold L^{\be_j}_N$.  Thus $\bold E^{\be_j}_N(z)=\bold M(z)e^{\pm\frac{i\Omega_{j,N}}{2}\sigma_3}\bold L^{\be_j}_N(z)^{-1}$ has no jumps in $D(\be_j,\ep)$.  The only other possible singularity for $\bold E^{\be_j}_N$ is at $\be_j$, and this singularity is at most a fourth root singularity, thus removable.  Thus, $\bold E^{\be_j}_N$ is analytic in $D(\be_j, \ep)$, and $\bold Q_N$ has the prescribed jumps.
We are left only to prove the matching condition (\ref{pm2}).  Using (\ref{pm6a}), one can check that, for $z$ in each of the sectors of analyticity, $\Phi_{rv}(N^{2/3}\psi_{\be_j}(z))$ satisfies the following asymptotics as $N\to\infty$:
\begin{equation}\label{pm13}
\begin{aligned}
\Phi_{rv}\big(N^{2/3}\psi_{\be_j}(z)\big)&=\frac{1}{2\sqrt{\pi}}N^{-\frac{1}{6}\sigma_3}\psi_{\be_j}(z)^{-\frac{1}{4}\sigma_3}
\left[ \begin{pmatrix} 1& i \\ -1 & i \end{pmatrix} + \frac{\psi_{\be_j}(z)^{-3/2}}{48N}
\begin{pmatrix} -5 & 5i \\ -7 & -7i \end{pmatrix}+O(N^{-2})\right]\\
&\qquad \times e^{-\frac{2}{3}N\psi_{\be_j}(z)^{3/2}\sigma_3}
\end{aligned}
\end{equation}
where we always take the principal branch of $\psi_{\be_j}(z)^{3/2}$.  As such, $\psi_{\be_j}(z)^{3/2}$ is two-valued for $z\in (\be_j-\ep, \be_j)$, so that 
\begin{equation}\label{pm14}
\left[\frac{2}{3}\psi_{\be_j}(x)^{3/2}\right]_{\pm}=\mp \pi i\int_x^{\be_j} \rho(t)dt.
\end{equation}
Notice that, by (\ref{g9}),  
\begin{equation}\label{pm15}
2g_{\pm}(x)-V(x)=l\pm 2\pi i\int_x^{\be_q}\rho(t)dt=l\pm 2\pi i\int_x^{\be_j}\rho(t)dt \pm i\Om_j
\end{equation}
This implies that for $x \in (\be_j-\ep, \be_j)$,
\begin{equation}\label{pm16}
\begin{aligned}
&[2g_+(\be_j)-V(\be_j)]-[2g_+(x)-V(x)]=-2\pi i \int_x^{\be_j} \rho(t)dt\,, \\
&[2g_-(\be_j)-V(\be_j)]-[2g_-(x)-V(x)]=2\pi i \int_x^{\be_j} \rho(t)dt\,. 
\end{aligned}
\end{equation}
Combining these equations with (\ref{pm14}) gives
\begin{equation}\label{pm17}
\left[\frac{2}{3}\psi_{\be_j}(x)^{3/2}\right]_{\pm}=\frac{1}{2}\bigg[\big(2g_{\pm}(\be_j)-V(\be_j)\big)-\big(2g_{\pm}(x)-V(x)\big)\bigg].
\end{equation}
This equation can be extended into the upper and lower planes, respectively, giving
\begin{equation}\label{pm18}
\frac{2}{3}\psi_{\be_j}(z)^{3/2}=\frac{1}{2}\bigg[\big(2g_{\pm}(\be_j)-V(\be_j)\big)-\big(2g(z)-V(z)\big)\bigg] \quad \textrm{for} \ \pm \Im z >0.
\end{equation}
Since, by (\ref{pm15}), $2g_{\pm}(\be_j)-V(\be_j)=l\pm i\Om_j$, we get that
\begin{equation}\label{pm19}
\frac{2}{3}\psi_{\be_j}(z)^{3/2}=-g(z)+\frac{V(z)}{2}+\frac{l}{2}\pm \frac{i\Om_j}{2}
\end{equation}
for $\pm \Im z >0$.
Plugging (\ref{pm13}) and (\ref{pm19}) into (\ref{pm11}), we get
\begin{equation}\label{pm20}
\begin{aligned}
\bold U_N(z)&=\bold M(z)e^{\pm \frac{i\Om_{j,N}}{2}}\bold L_N^{\be_j}(z)^{-1}\frac{1}{2\sqrt{\pi}}N^{-\frac{1}{6}\sigma_3}\psi_{\be_j}(z)^{-\frac{1}{4}\sigma_3}
\bigg[ \begin{pmatrix} 1& i \\ -1 & i \end{pmatrix} + \frac{\psi_{\be_j}(z)^{-3/2}}{48N}
\begin{pmatrix} -5 & 5i \\ -7 & -7i \end{pmatrix}\\
&\quad +O(N^{-2})\bigg] \ e^{N(g(z)-\frac{V(z)}{2}-\frac{l}{2} \mp \frac{i\Om_j}{2})\sigma_3}e^{-N(g(z)-\frac{V(z)}{2}-\frac{l}{2})\sigma_3} \\
&=\bold M(z)\left[ \bold I + \frac{\psi_{\be_j}(z)^{-3/2}}{48N}\begin{pmatrix} 1 & 6ie^{\pm i\Om_{j,N}} \\  6ie^{\mp i\Om_{j,N}} & -1 \end{pmatrix}+O(N^{-2})\right] 
\end{aligned}
\end{equation}
for $\pm \Im z > 0$.
Thus we have that $\bold U_N$ satisfies conditions (\ref{pm0}), (\ref{pm1}), and (\ref{pm2}).

A similar construction gives the parametrix at the $\al_j$ for $j-1 \in \acal_v$.  Namely, if we let
\begin{equation}\label{pm21}
\psi_{\al_j}(z)=-\left\{\frac{3\pi}{2}\int_{\al_j}^z \rho(t)dt\right\}^{2/3},
\end{equation}
then $\psi_{\al_j}$ is analytic throughout $D(\al_j,\ep)$, real valued on the real line, and has negative derivative at $\al_j$.  Close to $\al_j$, the jumps $j_Q$ become
\begin{equation}\label{pm21a}
j_Q(z)=
\left\{
\begin{aligned}
&\begin{pmatrix} 1 & 1 \\ 0 & 1 \end{pmatrix} \quad \textrm{for} \ z\in (\al_j-\ep,\al_j) \\
&\begin{pmatrix} 1& 0 \\ -1 & 1 \end{pmatrix}\quad \textrm{for} \ z\in (\al_j,\al_j+i\ep) \\
&\begin{pmatrix} 1& 0 \\ 1 & 1 \end{pmatrix}\quad \textrm{for} \ z\in (\al_j,\al_j-i\ep) \\
&\begin{pmatrix} 0 & 1 \\ -1 & 0 \end{pmatrix} \quad \textrm{for} \ z\in (\al_j,\al_j+\ep),
\end{aligned}\right.
\end{equation}
where orientation is taken left to right on horizontal contours, and up to down on vertical contours according to Figure \ref{sigma_S}.  After the change of variables $\psi_{\al_j}$ (and a slight deformation of vertical contours), these jumps become the following jumps close to the origin:
\begin{equation}\label{pm21b}
j_Q\big(\psi_{\al_j}(z)\big)=\left\{
\begin{aligned}
&\begin{pmatrix} 0 & 1 \\ -1 & 0 \end{pmatrix} \quad \textrm{for} \ \psi_{\al_j}(z)\in (-\ep,0) \\
&\begin{pmatrix} 1 & 0 \\ 1 & 1 \end{pmatrix}\quad \textrm{for} \ \psi_{\al_j}(z)\in (0,i\ep) \\
&\begin{pmatrix} 1 & 0 \\ -1 & 1 \end{pmatrix}\quad \textrm{for} \ \psi_{\al_j}(z)\in (0,-i\ep) \\
&\begin{pmatrix} 1 & 1 \\ 0 & 1 \end{pmatrix} \quad \textrm{for} \ \psi_{\al_j}(z)\in (0,\ep),
\end{aligned}\right.
\end{equation}
where orientation is taken right to left on horizontal contours, and down to up on vertical contours.
These jump conditions are satisfied by the function
\begin{equation}\label{pm21c}
\Phi_{lv}(z)=\Phi_{rv}(z)\begin{pmatrix} 1 & 0 \\ 0 & -1 \end{pmatrix}.
\end{equation}
Then we can take 
\begin{equation}\label{pm22}
\bold U_N(z)=\bold E_N^{\al_j}(z)\Phi_{lv}\big(N^{2/3}\psi_{\al_j}(z)\big)e^{-N(g(z)-\frac{V(z)}{2}-\frac{l}{2})\sigma_3}
\end{equation}
for $z\in D(\al_j,\ep)$, where
\begin{equation}\label{pm23}
\begin{aligned}
\bold E_N^{\al_j}(z)&=\bold M(z)e^{\pm \frac{i\Om_{j-1,N}}{2}\sg_3}\bold L_N^{\al_j}(z)^{-1} \quad \textrm{for} \quad \pm \Im z >0, \\
\bold L_N^{\al_j}(z)&=\frac{1}{2\sqrt{\pi}}
\begin{pmatrix} N^{-1/6} \psi_{\al_j}^{-1/4}(z) & 0 \\ 0 & N^{1/6} \psi_{\al_j}^{1/4}(z)\end{pmatrix}
\begin{pmatrix} 1 & -i \\ -1 & -i \end{pmatrix} \\
\end{aligned}
\end{equation}
is an analytic prefactor.
Similar to (\ref{pm13}), we have that in each sector of analyticity, $\Phi_{lv}(N^{2/3}\psi_{\al_j}(z))$ satisfies
\begin{equation}\label{pm24}
\begin{aligned}
\Phi_{lv}\big(N^{2/3}\psi_{\al_j}(z)\big)&=\frac{1}{2\sqrt{\pi}}
N^{-\frac{1}{6}\sigma_3}\psi_{\al_j}(z)^{-\frac{1}{4}\sigma_3}
\bigg[ \begin{pmatrix} 1& -i \\ -1 & -i \end{pmatrix} + \frac{\psi_{\al_j}(z)^{-3/2}}{48N}
\begin{pmatrix} -5 & -5i \\ -7 & 7i \end{pmatrix}+O(N^{-2})\bigg]\\
&\qquad \times e^{-\frac{2}{3}N\psi_{\al_j}(z)^{3/2}\sigma_3}.
\end{aligned}
\end{equation}
Once again, we have that, for $x\in (\al_j,\al_j+\ep)$, $\psi_{\al_j}(x)^{3/2}$ takes limiting values from above and below, so that
\begin{equation}\label{pm25}
\left[\frac{2}{3}\psi_{\al_j}(x)^{3/2}\right]_{\pm}=\pm \pi i\int_{\al_j}^x \rho(t)dt.
\end{equation}
In analogue to (\ref{pm17}), we have
\begin{equation}\label{pm25a}
\frac{2}{3}\psi_{\al_j}(z)^{3/2}=\frac{1}{2}\bigg[\big(2g_{\pm}(\al_j)-V(\al_j)\big)-\big(2g(z)-V(z)\big)\bigg]
 \quad \textrm{for} \ \pm \Im z >0.
\end{equation}
Since, by (\ref{pm15}), $2g_{\pm}(\al_j)-V(\al_j)=l\pm\pi i$, we get that
\begin{equation}\label{pm26}
\frac{2}{3}\psi_{\al_j}(z)^{3/2}=-g(z)+\frac{V(z)}{2}+\frac{l}{2} \pm\frac{i\Om_{j-1}}{2} \quad \textrm{for} \quad \pm \Im z >0.
\end{equation}
Plugging (\ref{pm26}) into (\ref{pm22}) and (\ref{pm24}) gives, as $N\to \infty$,
\begin{equation}\label{pm27}
\begin{aligned}
\bold U_N(z)&=\bold M(z)\bold L_N^{\al_j}(z)^{-1}\frac{1}{2\sqrt{\pi}}
N^{-\frac{1}{6}\sigma_3}\psi_{\al_j}(z)^{-\frac{1}{4}\sigma_3}
\bigg[ \begin{pmatrix} 1& -i \\ -1 & -i \end{pmatrix} + \frac{\psi_{\al_j}(z)^{-3/2}}{48N}
\begin{pmatrix} -5 & -5i \\ -7 & 7i \end{pmatrix}\\
&\qquad +O(N^{-2})\bigg]  e^{N(g(z)-\frac{V(z)}{2}-\frac{l}{2} \mp \frac{i\Om_{j-1}}{2})\sigma_3}e^{-N(g(z)-\frac{V(z)}{2}-\frac{l}{2})\sigma_3} \\
&=\bold M(z)\left[ \bold I + \frac{\psi_{\al_j}(z)^{-3/2}}{48N}
\begin{pmatrix} 1 & -6ie^{\frac{i\Om_{j-1,N}}{2}\sg_3} \\  -6ie^{-\frac{i\Om_{j-1,N}}{2}\sg_3} & -1 \end{pmatrix}+O(N^{-2})\right] .
\end{aligned}
\end{equation}

\section{Parametrix at the band-saturated region end points}

We now consider small disks $D(\al_j,\ep)$ for $j-1 \in \acal_s$, and $D(\be_j,\ep)$ for $j \in \acal_s$, centered at the endpoints of bands which are adjacent to a saturated region.   Denote 
\begin{equation}\label{pmi0a}
\tilde{D}=\left(\underset{j-1 \in \acal_s}\bigcup D(\al_j,\ep)\right) \bigcup \left(\underset{j \in \acal_s}\bigcup D(\be_j,\ep)\right).
\end{equation}

  We will seek a local parametrix $\bold U_N(z)$ defined on $\tilde{D}$ such that
\begin{enumerate}
\item 
\begin{equation}\label{pmi0}
\bold U_N(z) \ \textrm{is analytic on} \ \tilde{D} \setminus \Sigma_S.
\end{equation}
\item
\begin{equation}\label{pmi1}
\bold U_{N+}(z)=\bold U_{N-}(z)j_S(z) \quad \textrm{for} \quad z\in \tilde{D} \cap \Sigma_S.
\end{equation}
\item
\begin{equation}\label{pmi2}
\bold U_N(z)=\bold M(z) \big(I+O(N^{-1})\big) \quad \textrm{uniformly for} \  z\in \partial \tilde{D}.
\end{equation}
\end{enumerate}
We first construct the parametrix near $\be_j$ for $j \in \acal_s$.  Let 
\begin{equation}\label{pmi3}
\bold U_N(z)=\tilde{\bold Q}_N(z)e^{\mp iN\pi z\sigma_3}e^{-N(g(z)-\frac{V(z)}{2}-\frac{l}{2})\sigma_3} \quad \textrm{for} \quad \pm \Im z > 0.
\end{equation}
Then the jumps for $\tilde{\bold Q}_N$ are
\begin{equation}\label{pmi4}
j_{\tilde{Q}}(z)=\left\{
\begin{aligned}
&\begin{pmatrix} 0 & 1 \\ -1 & 0 \end{pmatrix} \quad &\textrm{for}\quad z \in (\be_j-\ep,\be_j) \\
&\begin{pmatrix} -1 & 0 \\ -1 & -1 \end{pmatrix}\quad & \textrm{for}\quad z \in (\be_j,\be_j+\ep) \\
&\begin{pmatrix} 1 & -1 \\ 0 & 1 \end{pmatrix} \quad & \textrm{for}\quad z \in (\be_j,\be_j+i \ep) \\
&\begin{pmatrix} 1 & 1 \\ 0 & 1 \end{pmatrix}  \quad & \textrm{for}\quad z \in (\be_j,\be_j-i \ep)\\
\end{aligned}\right.
\end{equation}
where orientation is taken from left to right on horizontal contours, and down to up on vertical contours according to Figure \ref{sigma_S}.
We now take
\begin{equation}\label{pmi5}
\Phi_{rs}(z)=\left\{
\begin{aligned}
&\begin{pmatrix}y_2(z) & -y_0(z) \\ y_2'(z) & -y_0'(z) \end{pmatrix} 
\quad \textrm{for} \quad \arg z \in \left(0,\frac{\pi}{2}\right) \\
&\begin{pmatrix}y_2(z) & y_1(z) \\ y_2'(z) & y_1'(z) \end{pmatrix} 
\quad \textrm{for} \quad \arg z \in \left(\frac{\pi}{2},\pi \right) \\
&\begin{pmatrix}y_1(z) & -y_2(z) \\ y_1'(z) & -y_2'(z) \end{pmatrix} 
\quad \textrm{for} \quad \arg z \in \left(-\pi,-\frac{\pi}{2}\right) \\
&\begin{pmatrix}y_1(z) & y_0(z) \\ y_1'(z) & y_0'(z) \end{pmatrix} 
\quad \textrm{for} \quad \arg z \in \left(-\frac{\pi}{2},0 \right).
\end{aligned}
\right.
\end{equation}
Then $\Phi_{rs}(z)$ solves a RHP similar to that of $\tilde{\bold Q}_N$, 
but for jumps emanating from the origin rather than from $\be_j$.

Once again,
\begin{equation}\label{pmi8}
\psi_{\be_j}(z)=-\left\{\frac{3\pi}{2}\int_z^{\be_j}\left(1-\rho(t)\right)dt\right\}^{2/3}
\end{equation}
extends to a conformal map from $D(\be_j, \ep)$ onto a convex neighborhood of the origin, with 
\begin{equation}\label{pmi9}
\psi_{\be_j}(\be_j)=0 \quad ; \quad \psi_{\be_j}'(\be_j)>0,
\end{equation}
Again, we can slightly deform the vertical pieces of the contour $\Sigma_S$ close to $\be_j$, so that
\begin{equation}\label{pm10bbb}
\psi_{\be_j}\big\{D(\be_j,\ep) \cap \Sigma_S\big\}=(-\ep,\ep) \cup (-i\ep,i\ep)
\end{equation}
We now take 
\begin{equation}\label{pmi11}
\tilde{\bold Q}_N(z)=\bold E_N^{\be_j}(z)\Phi_{rs}\big(N^{2/3}\psi_{\be_j}(z)\big)
\end{equation}
where
\begin{equation}\label{pmi12}
\begin{aligned}
&\bold E^{\be_j}_N(z)=\bold M(z)e^{\pm\frac{i\Omega_{j,N}}{2}\sigma_3}\bold L^{\be_j}_N(z)^{-1} \quad \textrm{for} \quad \pm \Im z \geq 0, \\
&\bold L^{\be_j}_N(z)=\frac{1}{2\sqrt{\pi}}\begin{pmatrix} N^{-1/6} \psi_{\be_j}^{-1/4}(z) & 0 \\ 0 & N^{1/6} \psi_{\be_j}^{1/4}(z)\end{pmatrix}\begin{pmatrix} 1 & i \\ 1 & -i \end{pmatrix},
\end{aligned}
\end{equation}
and we take the principal branch of $ \psi_{\be_j}^{1/4}$.  
%$\bold U_N$ then becomes
%\begin{equation}\label{pmi13}
%\bold U_N(z)=\left\{
%\begin{aligned}
%&\bold M(z)e^{\frac{i\Omega_{j,N}}{2}\sigma_3}\bold L^{\be_j}_N(z)^{-1}\Phi_{rs}\big(N^{2/3}\psi_{\be_j}(z)\big)e^{-iN\pi z\sigma_3}e^{-N(g(z)-\frac{V(z)}{2}-\frac{l}{2})\sigma_3}, \\ & \hskip 5cm \textrm{for} \quad \Im z > 0\\
%&\bold M(z)e^{-\frac{i\Omega_{j,N}}{2}\sigma_3}\bold L^{\be_j}_N(z)^{-1}\Phi_{rs}\big(N^{2/3}\psi_{\be_j}(z)\big)e^{iN\pi z\sigma_3}e^{-N(g(z)-\frac{V(z)}{2}-\frac{l}{2})\sigma_3}, \\ &\hskip 5cm \textrm{for}  \quad \Im z < 0 .\\
%\end{aligned}\right.
%\end{equation}
The function $\Phi_{rs}(N^{2/3}\psi_{\be_j}(z))$ has the jumps $j_S$.  Similar to the prefactor $\bold E_N^{\be_j}$ at band-void end-points, the prefactor $\bold E^{\be_j}_N$ is analytic in $D(\be_j, \ep)$, thus does not change these jumps. 

We now check that $\bold U_N$ satisfies the matching condition (\ref{pmi2}).  The large $N$ asymptotics of $\Phi_{rs}(N^{2/3}\psi_{\be_j}(z))$ are given in the different regions of analyticity as follows:
\begin{equation}\label{pmi17}
\begin{aligned}
\Phi_{rs}\big(N^{2/3}\psi_{\be_j}(z)\big)&=\frac{1}{2\sqrt{\pi}}
N^{-\frac{1}{6}\sigma_3}\psi_{\be_j}(z)^{-\frac{1}{4}\sigma_3}
\bigg[\pm \begin{pmatrix} -i & -1 \\ -i & 1 \end{pmatrix}
\pm \frac{\psi_{\be_j}(z)^{-3/2}}{48N}\begin{pmatrix} -5i & 5 \\ 7i & 7 \end{pmatrix}\\
&\quad +O(N^{-2})\bigg]
e^{\frac{2}{3}N\psi_{\be_j}(z)^{3/2}\sigma_3}
\hspace{1 cm} \textrm{for} \quad \pm \Im z>0, 
\end{aligned}
\end{equation}
where we always take the principal branch of $\psi_{\be_j}(z)^{3/2}$.  As such, $\psi_{\be_j}(z)^{3/2}$ is
two-valued for $x \in (\be_j-\ep, \be_j)$, so that
\begin{equation}\label{pmi18}
\left[\frac{2}{3}\psi_{\be_j}(x)^{3/2}\right]_{\pm}=\mp\pi i\int_x^{\be_j}\left(1-\rho(t)\right)dt
=\mp \pi i(\be_j-x)\pm\pi i\int_x^{\be_j}\rho(t)dt.
\end{equation}
From (\ref{g9}) we have that
\begin{equation}\label{pmi19}
2g_\pm(x)-V(x)=l\pm 2\pi i \int_x^{\be_q} \rho(t)dt=l\pm 2\pi i \int_x^{\be_j} \rho(t)dt \pm i\Om_j \mp 2\pi i \be_j
\end{equation}
for $x\in (\be_j-\ep,\be_j)$.  These equations imply that
\begin{equation}\label{pmi19a}
\big(2g_\pm(x)-V(x)\big)-\big(2g_\pm(\be_j)-V(\be_j)\big)=\pm 2\pi i \int_x^{\be_j} \rho(t)dt.
\end{equation}
We can therefore write (\ref{pmi18}) as
\begin{equation}\label{pmi19b}
\left[\frac{2}{3}\psi_{\be_j}(x)^{3/2}\right]_{\pm}=\mp\pi i(\be_j-x)+\frac{1}{2}\bigg[\big(2g_{\pm}(x)-V(x)\big)-\big(2g_{\pm}(\be_j)-V(\be_j)\big)\bigg].
\end{equation}
We can extend these equations into the upper and lower half-plane, respectively, obtaining
\begin{equation}\label{pmi20}
\frac{2}{3}\psi_{\be_j}(z)^{3/2}=
\mp \pi i(\be_j-z)+\frac{1}{2}\bigg[\big(2g(z)-V(z)\big)-\big(2g_\pm(\be_j)-V(\be_j)\big)\bigg] \quad \textrm{for} \quad \pm \Im z>0. \\
\end{equation}
Using (\ref{pmi19}) at $x=\be_j$, we can write
\begin{equation}\label{pmi22}
\frac{2}{3}\psi_{\be_j}(z)^{3/2}=g(z)-\frac{V(z)}{2}-\frac{l}{2}\pm \pi iz \mp \frac{i(\Omega_{j,N}-\pi)}{2N} \quad \textrm{for} \quad \pm \Im z>0. 
\end{equation}
Plugging (\ref{pmi17}) and (\ref{pmi22}) into (\ref{pmi11}) gives
\begin{equation}\label{pmi22a}
\begin{aligned}
\bold U_N(z)&=\bold M(z)e^{\frac{i\Omega_{j,N}}{2}\sigma_3}\bold L_N^{\be_j}(z)^{-1}\frac{1}{2\sqrt{\pi}}
N^{-\frac{1}{6}\sigma_3}\psi_{\be_j}(z)^{-\frac{1}{4}\sigma_3}\\
&\qquad  \times \bigg[\pm \begin{pmatrix} -i & -1 \\ -i & 1 \end{pmatrix}\pm \frac{\psi_{\be_j}(z)^{-3/2}}{48N}
\begin{pmatrix} -5i & 5 \\ 7i & 7 \end{pmatrix}+O(N^{-2})\bigg] \\
&\qquad  \times e^{N(g(z)-\frac{l}{2}-\frac{V(z)}{2})\sigma_3}e^{\pm iN\pi z\sigma_3}e^{\mp\frac{i\Omega_{j,N}}{2}\sigma_3}e^{\pm\frac{i\pi}{2}\sigma_3}e^{\mp iN\pi z\sigma_3}e^{-N(g(z)-\frac{V(z)}{2}-\frac{l}{2})\sigma_3} \\
&=\bold M(z)\left[I+\frac{\psi_{\be_j}(z)^{-3/2}}{48N}
\begin{pmatrix} -1 & -6ie^{\pm i\Omega_{j,N}} \\ -6ie^{\mp i\Omega_{j,N}} & 1 \end{pmatrix}+O(N^{-2})\right] 
\quad \textrm{for} \ \pm \Im (z) >0.
\end{aligned}
\end{equation}

We can make a similar construction near $\al_j$ for $j-1 \in \acal_s$.  Let
\begin{equation}\label{pmi24}
\psi_{\al_j}(z)=-\left\{\frac{3\pi}{2}\int_{\al_j}^z\left(1-\rho(t)\right)dt\right\}^{2/3}.
\end{equation}
This function is analytic in $D(\al_j,\ep)$ and has negative derivative at $\al_j$, thus $\Im z$ and $\Im \psi_{\al_j}(z)$ have opposite signs for $z \in D(\al_j,\ep)$.  
Then the jumps for $\tilde{\bold Q}_N$ are
\begin{equation}\label{pmi24a}
j_{\tilde{Q}}(z)=\left\{
\begin{aligned}
&\begin{pmatrix} 0 & 1 \\ -1 & 0 \end{pmatrix} \quad &\textrm{for}\quad z \in (\al_j,\al_j+\ep) \\
&\begin{pmatrix} -1 & 0 \\ -1 & -1 \end{pmatrix}\quad & \textrm{for}\quad z \in (\al_j-\ep,\al_j) \\
&\begin{pmatrix} 1 & -1 \\ 0 & 1 \end{pmatrix} \quad & \textrm{for}\quad z \in (\al_j,\al_j+i\ep) \\
&\begin{pmatrix} 1 & 1 \\ 0 & 1 \end{pmatrix}  \quad & \textrm{for}\quad z \in (\al_j,\al_j-i \ep),\\
\end{aligned}\right.
\end{equation}
where the contour is oriented from left to right on horizontal segments and up to down on vertical segments according to Figure \ref{sigma_S}.  After a slight deformation of the vertical contours and the change of variables $\psi_{\al_j}$, these jumps become the following jumps close to the origin:
\begin{equation}\label{pmi24b}
j_{\tilde{Q}}(\psi_{\al_j}(z))=\left\{
\begin{aligned}
&\begin{pmatrix} 0 & 1 \\ -1 & 0 \end{pmatrix} \quad &\textrm{for}\quad \psi_{\al_j}(z) \in (-\ep,0) \\
&\begin{pmatrix} -1 & 0 \\ -1 & -1 \end{pmatrix}\quad & \textrm{for}\quad \psi_{\al_j}(z) \in (0,\ep) \\
&\begin{pmatrix} 1 & -1 \\ 0 & 1 \end{pmatrix} \quad & \textrm{for}\quad \psi_{\al_j}(z) \in (-i \ep,0) \\
&\begin{pmatrix} 1 & 1 \\ 0 & 1 \end{pmatrix}  \quad & \textrm{for}\quad \psi_{\al_j}(z) \in (0,i \ep),\\
\end{aligned}\right.
\end{equation}
where the contour is oriented from right to left on horizontal segments and down to up on vertical segments.  These jump conditions are satisfied by the function 
\begin{equation}\label{pmi25c}
\Phi_{ls}(z)=\Phi_{rs}(z)\begin{pmatrix} 1 & 0 \\ 0 & -1 \end{pmatrix}.
\end{equation}
Then we can take for $z\in D(\al_j,\ep)$, 
\begin{equation}\label{pmi25}
\bold U_N(z)=\bold M(z)e^{\frac{\pm i\Omega_{j,N}}{2}\sigma_3}\bold L_N^{\al_j}(z)^{-1}\Phi_{ls}(N^{2/3}\psi_{\al_j}(z))e^{\mp iN\pi z\sigma_3}e^{-N(g(z)-\frac{V(z)}{2}-\frac{l}{2})\sigma_3} \quad \textrm{for} \quad \pm \Im z > 0,
\end{equation}
where
\begin{equation}\label{pmi26}
\bold L_N^{\al_j}(z)=\frac{1}{2\sqrt{\pi}}
\begin{pmatrix} N^{-1/6} \psi_{\al_j}^{-1/4}(z) & 0 \\ 0 & N^{1/6} \psi_{\al_j}^{1/4}(z)\end{pmatrix}
\begin{pmatrix} -1 & i \\ -1 & -i \end{pmatrix}.
\end{equation}
We once again have 
\begin{equation}\label{pmi27}
\begin{aligned}
\Phi_{ls}\big(N^{2/3}\psi_{\al_j}(z)\big)&=\frac{1}{2\sqrt{\pi}}N^{-\frac{1}{6}\sigma_3}\psi_{\al_j}(z)^{-\frac{1}{4}\sigma_3}
\bigg[\pm \begin{pmatrix} -i & 1 \\ -i & -1 \end{pmatrix} \pm \frac{\psi_{\al_j}(z)^{-3/2}}{48N}
\begin{pmatrix} -5i & -5 \\ 7i & -7 \end{pmatrix}\\
&\quad +O(N^{-2})\bigg]
e^{\frac{2}{3}N\psi_{\al_j}(z)^{3/2}\sigma_3} 
\hspace{5mm} \textrm{for} \; \pm  \Im \psi_{\al_j}(z)>0 \; (\textrm{so} \; \mp \Im z > 0),\\
\end{aligned}
\end{equation}
and for $z \in D(\al_j,\ep)$, 
\begin{equation}\label{pmi28}
\frac{2}{3} \psi_{\al_j}^{3/2}(z)=\pm i \pi z +g(z)-\frac{V(z)}{2}-\frac{l}{2} \mp \frac{i(\Omega_{j,N}-\pi)}{2N} \quad \textrm{for} \quad \pm \Im z >0.
\end{equation}
Combining (\ref{pmi25}), (\ref{pmi27}), and (\ref{pmi28}) gives
\begin{equation}\label{pmi29}
\begin{aligned}
\bold U_N(z)&=\bold M(z)e^{\frac{\pm i\Omega_{j,N}}{2}\sigma_3}
\bold L_N^{\al_j}(z)^{-1}\frac{1}{2\sqrt{\pi}}N^{-\frac{1}{6}\sigma_3}
\psi_{\al_j}(z)^{-\frac{1}{4}\sigma_3}\\
&\qquad \times\left[\pm \begin{pmatrix} i & -1 \\ i & 1 \end{pmatrix} \pm \frac{\psi_{\al_j}(z)^{-3/2}}{48N}
\begin{pmatrix} 5i & 5 \\ -7i & 7 \end{pmatrix} +O(N^{-2})\right] \\
&\qquad \times e^{\pm i N \pi z\sigma_3} e^{N(g(z)-\frac{V(z)}{2}-\frac{l}{2})\sigma_3}
e^{\mp \frac{i\Omega_{j,N}}{2}\sigma_3}e^{\pm \frac{i\pi}{2}\sigma_3}e^{\mp iN\pi z\sigma_3}
e^{-N(g(z)-\frac{V(z)}{2}-\frac{l}{2})\sigma_3} \\
&=\bold M(z)e^{\frac{\pm i\Omega_{j,N}}{2}\sigma_3}
\left[\bold I+\frac{\psi_{\al_j}(z)^{-3/2}}{48N}\begin{pmatrix} -1 & 6i \\ 6i & 1 \end{pmatrix} +O(N^{-2})\right]
e^{\mp \frac{i\Omega_{j,N}}{2}\sigma_3} \\
&=\bold M(z)\left[\bold I+\frac{\psi_{\al_j}(z)^{-3/2}}{48N}
\begin{pmatrix} -1 & 6ie^{\pm i\Omega_{j,N}} \\ 6ie^{\mp i\Omega_{j,N}} & 1 \end{pmatrix} +O(N^{-2})\right] 
\quad \textrm{for} \ \pm \Im z >0.
\end{aligned}
\end{equation}

\section{The third and final transformation of the RHP}\label{tt}

We now consider the contour $\Sigma_X$, which consists of the circles $\partial D(\al_j, \ep)$, and $\partial D(\be_j, \ep)$, for $j= 1,\dots q$, all oriented counterclockwise, together with the parts of 
$\Sigma_S\setminus\big(\underset{j}\bigcup [\al_j,\be_j]\big)$ 
which lie outside of the disks $D(\al, \ep), \ D(\al', \ep), \ D(\be', \ep),$ and $D(\be, \ep)$, see Fig.~4.

%%%%%%%%%%%%% Fig.  %%%%%%%%%%%%%%
\begin{center}
 \begin{figure}[h]
\begin{center}
   \scalebox{0.8}{\includegraphics{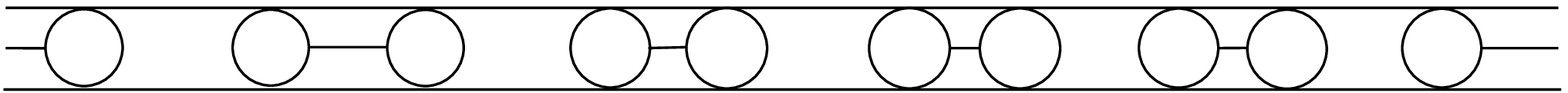}}
\end{center}
        \caption[contour $\Sigma_X$.]{The contour $\Sigma_X$ arising from the hypothetical equilibrium measure shown in Figure \ref{fig1}.}
\label{sigma_X}
    \end{figure}
\end{center}
%%%%%%%%%%%%%%%%%%%%%%%%%%%%%%%%%%%

We let
\begin{equation}\label{tt1}
\bold X_N(z)=\left\{
\begin{aligned}
&\bold S_N(z) \bold M(z)^{-1} \quad \textrm{for} \ z \ \textrm{outside the disks }  D(\al_j, \ep), \ D(\be_j, \ep) \\
&\bold S_N(z) \bold U_N(z)^{-1} \quad \textrm{for} \ z \ \textrm{inside the disks }  D(\al_j, \ep), \ D(\be_j, \ep). \\
\end{aligned}\right.
\end{equation}
Then $\bold X_N(z)$ solves the following RHP:
\begin{enumerate}
\item   
$\bold X_N(z)$ is analytic on $\C \setminus \Sigma_X$.
\item
$\bold X_N(z)$ has the jump properties
\begin{equation}\label{tt2}
\bold X_{N+}(x)=\bold X_{N-}(z)j_X(z)
\end{equation}
where
\begin{equation}\label{tt3}
j_X(z)=\left\{
\begin{aligned}
&\bold M(z)\bold U_N(z)^{-1} \quad \textrm{for} \ z \ \textrm{on the circles} \\
&\bold M(z)j_S\bold M(z)^{-1} \quad \textrm{otherwise}.
\end{aligned}\right.
\end{equation}
\item
As $z\to\infty$, 
\begin{equation}\label{tt4}
\bold X_N(z)\sim \I+\frac {\bold X_1}{z}+\frac {\bold X_2}{z^2}+\ldots
\end{equation}
\end{enumerate}
Additionally, we have that $j_X(z)$ is uniformly close to the identity in the following sense:
\begin{equation}\label{tt5}
j_X(z)=\left\{
\begin{aligned}
&\bold I+O(N^{-1}) \quad \textrm{uniformly on the circles} \\
&\bold I+O(e^{-C(z)N}) \quad \textrm{on the rest of} \ \Sigma_X, 
\end{aligned}\right.
\end{equation}
where $C(z)$ is a positive, continuous function satisfying (\ref{in3}).
If we set
\begin{equation}\label{tt6}
j_X^0(z)=j_X(z)-I,
\end{equation}
then (\ref{tt5}) becomes
\begin{equation}\label{tt7}
j_X^0(z)=\left\{
\begin{aligned}
&O(N^{-1}) \quad \textrm{uniformly on the circles} \\
&O(e^{-C(z)N}) \quad \textrm{on the rest of} \ \Sigma_X.
\end{aligned}\right.
\end{equation}
The solution to the RHP for $\bold X_N$ is based on the following lemma:
\begin{lem}
Suppose v(z) is a function on $\Sigma_X$ solving the equation
\begin{equation}\label{tt8}
v(z)=\bold I-\frac{1}{2\pi i} \int_{\Sigma_X} \frac{v(u)j_X^0(u)}{z_--u}du \quad \textrm{for} \ z\in\Sigma_X
\end{equation}
where $z_-$ means the value of the integral on the minus side of $\Sigma_X$.  Then 
\begin{equation}\label{tt9}
\bold X_N(z)=\bold I-\frac{1}{2\pi i} \int_{\Sigma_X} \frac{v(u)j_X^0(u)}{z-u}du \quad \textrm{for} \ z\in\C\setminus\Sigma_X
\end{equation}
solves the RHP for $\bold X_N$.
\end{lem}
The proof of this lemma is immediate from the jump property of the Cauchy transform.  By assumption
\begin{equation}\label{tt12}
\bold X_{N-}(z)=v(z) 
\end{equation}
and the additive jump of the Cauchy transform gives
\begin{equation}\label{tt13}
\bold X_{N+}(z)-\bold X_{N-}(z)=v(z)j_X^0(z)=\bold X_{N-}(z)j_X^0(z)
\end{equation}
thus $\bold X_{N+}(z)=\bold X_{N-}(z)j_X(z)$.  Asymptotics at infinity are given by (\ref{tt9}).

The solution to equation (\ref{tt8}) is given by a series of perturbation theory.  Namely, the solution is
\begin{equation}\label{tt14}
v(z)=\bold I+\sum_{k=1}^\infty v_k(z)
\end{equation}
where 
\begin{equation}\label{tt15}
v_k(z)=-\frac{1}{2\pi i} \int_{\Sigma_X} \frac{v_{k-1}(u)j_X^0(u)}{z-u}du \quad ; \quad v_0(z)=I.
\end{equation}
This function clearly solves (\ref{tt8}) provided the series converges, which it does, for sufficiently large $N$.  Indeed, by (\ref{tt5}), 
\begin{equation}\label{tt16}
|v_k(z)| \leq \left(\frac{C}{N}\right)^k\frac{1}{1+|z|} \quad \textrm{for some constant} \ C>0\,,
\end{equation}
thus the series (\ref{tt14}) is dominated by a convergent geometric series and thus converges absolutely.  This in turn gives
\begin{equation}\label{tt17}
\bold X_N(z)=\bold I+\sum_{k=1}^\infty\bold X_{N,k}(z)
\end{equation}
where
\begin{equation}\label{tt18}
\bold X_{N,k}(z)=-\frac{1}{2\pi i}\int_{\Sigma_X}\frac{v_{k-1}(u)j_X^0(u)}{z-u}du.
\end{equation}
In particular, this implies that 
\begin{equation}\label{tt19}
\bold X_N \sim \bold I + O\left(\frac{1}{N(|z|+1)}\right) \quad  \textrm{as} \quad N \to \infty 
\end{equation}
uniformly for $z\in \C \setminus \Sigma_X$.

\section{Proof of theorems \ref{coeff}-\ref{band_sat-reg}}\label{proof}
The transformations (\ref{red22}), (\ref{ft1}), (\ref{st1}), (\ref{tt1}) give that, for $z$ bounded away from the real line,
\begin{equation}\label{as1}
\bold P_N(z)=\bold K_N^{-1}e^{\frac{Nl}{2}\sg_3} \bold X_N(z) \bold M(z) e^{N(g(z)-\frac{l}{2})\sg_3}\bold K_N,
\end{equation}
and for $z$ close to the real line but bounded away from the support of the equilibrium measure, 
\begin{equation}\label{as1a}
\bold P_N(z)=\bold K_N^{-1}e^{\frac{Nl}{2}\sg_3} \bold X_N(z) \bold M(z) e^{N(g(z)-\frac{l}{2})\sg_3}\bold K_N \bold D_{\pm}^u(z)^{-1} \quad \textrm{for} \quad \pm \Im z \geq 0.
\end{equation}
Expanding (\ref{as1}) or (\ref{as1a}), we get that
\begin{equation}\label{as1b}
P_N(z)=[\bold P_N(z)]_{11}=e^{Ng(z)}\left([\bold M]_{11}[\bold X]_{11}+[\bold M]_{21}[\bold X]_{12}\right)
\end{equation}
which, along with (\ref{tt19}), proves Theorem \ref{voids}.

The proof of Theorem \ref{coeff} requires only the formulae (\ref{IP2}), (\ref{IP6}), and (\ref{IP7}), and a straightforward large-$z$ expansion of equation (\ref{as1}). 

%For any interval $J \in \R$, and any $\de > 0$, introduce here the notation
%\begin{equation}\label{as2}
%K_J^\de = \bigcup_{w\in J} \{z \in \C \ \ \textrm{such that} \ \ |z-w| \leq \de\}.
%\end{equation}

Similar to (\ref{as1}), we have that, for any interval $J$ which is contained in and bounded away from the endpoints of a band, in some neighborhood of $J$, we have
\begin{equation}\label{as3}
\bold P_N(z)=\left \{
\begin{aligned}
&\bold K_N^{-1}e^{\frac{Nl}{2}\sg_3} \bold X_N(z) \bold M(z) j_+(z) e^{N(g(z)-\frac{l}{2})\sg_3}\bold K_N \bold D_+^u(z)^{-1} \quad &\textrm{for} \quad \Im z &\geq 0    \\
&\bold K_N^{-1}e^{\frac{Nl}{2}\sg_3} \bold X_N(z) \bold M(z) j_-^{-1}(z) e^{N(g(z)-\frac{l}{2})\sg_3}\bold K_N \bold D_-^u(z) ^{-1}\quad &\textrm{for} \quad \Im z &\leq 0,
\end{aligned}\right.
\end{equation}
Expanding the left side of this equation for $\Im z \geq 0$, utilizing (\ref{g9}), and taking limits as $z$ approaches the real line, we get that
\begin{equation}\label{as3a}
P_N(x)=[\bold P_N(x)]_{11}=e^{\frac{N}{2}(V(x)+l)}\left(e^{iN\pi \phi(x)}[\bold M_{11}]_+(x)+e^{-iN\pi \phi(x)}[\bold M_{12}]_+(x)+O(N^{-1})\right)
\end{equation}
where $\phi(x)$ is as defined in (\ref{mr7}), and the $+$ subscript indicates the limiting value from the upper half plane.  Notice that $[\bold M_{12}]_+=[\bold M_{11}]_-$ in this region, and that $\bold M_{11}(\overline{z})=\overline{\bold M_{11}(z)}$.  This implies that $[\bold M_{12}]_+(x)=\overline{[\bold M_{11}]_+(x)}$, and thus we can write (\ref{as3a}) as
\begin{equation}\label{as3b}
P_N(x)=e^{\frac{N}{2}(V(x)+l)}\left(e^{iN\pi \phi(x)}[\bold M_{11}]_+(x)+\overline{e^{iN\pi \phi(x)}[\bold M_{11}]_+(x)}+O(N^{-1})\right),
\end{equation}
which proves Theorem \ref{bands}.

For any interval $J$ which is contained in and bounded away from the endpoints of a saturated region, in some neighborhood of $J$, we have
\begin{equation}\label{as4}
\bold P_N(z)=
\bold K_N^{-1}e^{\frac{Nl}{2}\sg_3} \bold X_N(z) \bold M(z) \bold A_{\pm}^{-1}(z) e^{N(g(z)-\frac{l}{2})\sg_3}\bold K_N \bold D_{\pm}^l(z)^{-1} \quad \textrm{for} \quad \pm \Im z> 0.   
\end{equation}
Notice that in this region, we can write
\begin{equation}\label{as5}
g_{\pm}(x)=L(x)\pm \frac{i\Om_j}{2} \mp i\pi x
\end{equation}
where $L(x)$ is defined in (\ref{mr9}).  Notice also that $2g_{\pm}(x)-V(x)-l$ has positive real part.  Expanding (\ref{as4}) for $\Im z >0$ and taking the limit as $z$ approaches the real line gives
\begin{equation}\label{as6}
\begin{aligned}
\bold P_N(x)_{11}&=e^{Ng_+(x)}\left[(1-e^{2\pi i N x})(\bold M_{11} \bold X_{11} +\bold M_{21} \bold X_{12}) + e^{-N(2g_+(x)-V(x)-l)}(\bold M_{12} \bold X_{11} +\bold M_{22} \bold X_{12})\right] \\
&=e^{Ng_+(x)}\left[(1-e^{2\pi i N x})(\bold M_{11} \bold X_{11} +O(N^{-1})) + O(e^{-N\de})\right] \\
&=e^{NL(x)}\left[-2i \sin (\pi N x)e^{\frac{iN\Om_j}{2}}[\bold M_{11}]_+(x) (1 +O(N^{-1})) + O(e^{-N\de})\right],
\end{aligned}
\end{equation}
which proves Theorem \ref{satreg}.

Similarly, at the turning points $\al_j$ and $\be_j$, explicit formulae can be written for $\bold P_N$ in terms of explicit transformations in each sector of analyticity of the local parametrix.  From these formulae and the properties of the $g$-function, Theorems \ref{band_void} and \ref{band_sat-reg} are almost immediate, with Theorem \ref{band_sat-reg} also requiring the identities (see, e.g. \cite{Ol})
\begin{equation}\label{as7}
\begin{aligned}
y_1(z)&=-\frac{1}{2}\big(\Ai(z)-i \Bi (z)\big) \\
y_2(z)&=-\frac{1}{2}\big(\Ai(z)+i \Bi (z)\big). 
\end{aligned}
\end{equation}

\end{document}